\documentclass[11pt,a4paper]{article}
\usepackage{jheppub}
\usepackage[T1]{fontenc}
\usepackage{amssymb}
\usepackage{mathtools}
\usepackage{amsmath}
\usepackage{bm}
\usepackage{stackengine}
\usepackage{scalerel}
\usepackage{framed}
\usepackage{mathtools}
\usepackage{amsthm}
\newtheorem{lemma}{Lemma}
\usepackage{color}
\usepackage{bbold}
\usepackage{bm}
\usepackage{latexsym}
\usepackage{braket}
\usepackage{slashed}
\usepackage{graphicx}
\usepackage{float} 
\usepackage{graphicx}
\usepackage{color} 
\usepackage{url}
\usepackage{tcolorbox}
\tcbuselibrary{fitting,raster,skins,xparse}
\usepackage{caption}
\usepackage{subcaption}

\usepackage{multirow}%
\numberwithin{equation}{section}
\usepackage{algpseudocode}

\makeatletter
\newcommand{\doublewidetilde}[1]{{%
  \mathpalette\double@widetilde{#1}%
}}
\newcommand{\double@widetilde}[2]{%
  \sbox\z@{$\m@th#1\widetilde{#2}$}%
  \ht\z@=.9\ht\z@
  \widetilde{\box\z@}%
}
\makeatother

\usepackage{array,multirow}
\usepackage{booktabs}
\title{\boldmath Statistical features of quantum chaos using the Krylov operator complexity \unboldmath}

\author{Zhuoran Li}
\author[1]{and Wei Fan~\note{Corresponding author. (The administrative policy of our university requires us to put students' name in the first place, so we did not follow the convention of arranging the names in the alphabet order. We apologize for this. ) }}
\affiliation{Department of Physics, School of Science, Jiangsu University of Science and Technology,
Zhenjiang, 212100, China}

\emailAdd{zrl2lzr@gmail.com, fanwei@just.edu.cn}

\abstract{
We study the statistical properties of Lanczos coefficients over an ensemble of random initial operators generating the Krylov space. We propose two statistical quantities that are important in characterizing the complexity: the average correlation matrix $\langle x_{i} x_{j}\rangle$ of Lanczos coefficients and the resulting distribution of the variance of Lanczos coefficients. Their resulting statistics are  the Wishart distribution and the (rescaled) chi-square distribution respectively, which are independent of the distributions of initial operators and become the normal distribution in the case of large matrix size. As a numerical example, we use the typical billiard system  with an integrability-breaking term and choose samples of random initial operators from given probability distributions (GOE, GUE and the uniform distribution).  It agrees with the phenomenological analysis and further interesting behaviors are obtained, which  indicates a consistent connection between RMT, Anderson localization and Krylov complexity. 
}

\begin{document}

\maketitle
\flushbottom

\section{Introduction}
\label{sec:intro}

The physics of complexity is important for understanding many different phenomena, ranging from condensed matter physics to black hole physics. Especially, the operator complexity captures  the spread of an operator under time evolution and is an important indicator in the analysis of chaos.  Recently the spread of 
an operator is explored in the Krylov space~\cite{Parker:2018yvk}, where the dynamics depends on the expansion coefficients  of the operator in the Krylov basis and on the Lanczos coefficients  generated in the process of constructing this Krylov basis. There the Krylov complexity (K-complexity) is constructed  from the expansion coefficients and  is proposed as an upper bound on the operator complexity, including the famous bound $\lambda_L \leq 2 \pi k_B T / \hbar$~\cite{Maldacena:2015waa} of the growth of the out-of-time-order correlation function  (OTOC)~\cite{1969JETP28:1200L}. 
%% 简短简介Krylov相关的一些著名方面, 这种某个模型中的Krylov，或者Krylov的某个方面的性质，就只提第一个开始研究的文章即可。%%【】提一下Toda chain的工作。

Interestingly, the operator dynamics in the Krylov space resembles the hopping dynamics on a 1d chain~\cite{Parker:2018yvk}, called the Krylov chain, with the hopping amplitudes being the Lanczos coefficients. Using this analog, the erratic behavior of Lanczos coefficients is analyzed~\cite{Trigueros:2021rwj,Rabinovici:2021qqt} in terms of the Anderson localization. The variance $\sigma^2$ of  Lanczos coefficients is introduced~\cite{Rabinovici:2021qqt} to explain the suppression of the late-time saturation value of K-complexity in the integrable XXZ spin chain.  This suggests that the variance $\sigma^2$ of  Lanczos coefficients may play an important role in distinguishing chaotic and nonchaotic behaviors. 

Then it leads to an interesting connection between K-complexity and Random Matrix Theory (RMT). By adding an integrability-breaking term to the XXZ spin chain,  the late-time saturation value of K-complexity is computed as the system changes from nonchaotic to chaotic behavior~\cite{Rabinovici:2022beu}. And it approaches the behavior of similar values computed with Hamiltonian drawn from RMT ensembles. Also, three different initial operators generating the Krylov space is randomly selected and compared there. In~\cite{Craps:2024suj}, the K-complexity is also studied by choosing multiple initial operators. 

Furthermore, a statistical correlation between the variance $\sigma^2$ of  Lanczos coefficients and the level statistics~\cite{Dyson:1962JMP,Gutzwiller:1971fy,Berry:1977RSPSA,Bohigas:1983er,Blumel:1990zz}  is reported in the billiard systems~\cite{Hashimoto:2023swv}. As the system Hamiltonian varies from nonchaotic to chaotic, the variance $\sigma^2$ of  Lanczos coefficients is computed for an initial momentum operator and the level statistics are computed as usual. The linear correlation coefficient between these two sequences of data are not zero. 

%%通过TBM引入的variance重要性，随机挑选算符. 
%%% 分析一些相关的研究. 【】如何恰当地引用手动选算符的文章

Along this line, an important question arises: \textit{ What would happen for the Lanczos coefficients if we were to perform statistics of initial operators over a probability distribution}. In this paper, we study this possibility by proposing two statistical quantities that have a deep connection with Anderson localization and RMT: the average correlation matrix $\langle x_{i} x_{j}\rangle$ of Lanczos coefficients and the resulting distribution of the variance $\sigma^2$. We qualitatively analyze their statistical properties based on RMT and on existing results in Krylov complexity. The average correlation matrix $\langle x_{i} x_{j}\rangle$ satisfies the Wishart distribution and the resulting distribution of $\sigma^2$ satisfies a rescaled  chi-square  distribution. This is independent of the initial operator distributions and approaches the normal distribution as the size of the operator matrix increases. For the purpose of characterizing operator complexity, the behavior of two quantities should be consistent with each other because of their deep connection with Anderson localization and RMT.  

As a numerical example, we use two typical billiard systems originating from a quantization of classical chaos: the Sinai billiard~\cite{BERRY1981163} and the Stadium billiard~\cite{benettin1978numerical,McDonald:1979zz,Casati:1980ytd}, equipped with an integrability-breaking term. The numerical results indeed verify our qualitative analysis and indicate the consistency between them, which implies an interesting connection among RMT, Anderson localization and Krylov complexity. Furthermore, for the Sinai billiard varing from nonchaotic to chaotic, these two quantities have distinct patterns that can capture this transition. In the nonchaotic regime the correlation matrix $\langle x_{i} x_{j}\rangle$ has a pattern that locally resembles the two-electron wave function of the Anderson localization~\cite{Weinmann1995h2eOF}, which disappears (or fading away by spreading out uniformly) in the chaotic regime. And the resulting distributions of $\sigma^2$ are nearly the same in the nonchaotic case, but split into two well-separated groups in the chaotic case. However, there is no such distinction for the Stadium billiard. In both examples, the behavior of these two quantities are consistent with each other.

For details of the numerical example, it is done in the energy representation and operators become matrices. Each time, a random initial operator $\mathcal{O}$ is sampled  out by picking the corresponding random matrix from the following distributions: the GOE,  the GUE,  the uniform distribution of complex numbers (UCP), the uniform distribution of pure real numbers (URE) and  of  pure imaginary numbers (UIM). Then the Lanczos coefficients are computed in the Krylov space associated with this operator, which give the correlation matrix $x_{i} x_{j}$ and the variance $\sigma^2$ of  Lanczos coefficients. Repeating this enough times,  
the average correlation matrix $\langle x_{i} x_{j}\rangle$ and the distribution of the variance $\sigma^2$ can be extracted for different samples of initial operators. Finally, these two statistical quantities are computed as the system changes from nonchaotic to chaotic.

The paper is organized as follows. In Section~\ref{sec:algorithm}, we briefly review the  relevant physics of the Krylov operator space. In Section~\ref{sec:statistics-property}, we  qualitatively analyze the statistical properties of $\langle x_{i} x_{j}\rangle$ and $\sigma^2$, and obtain the Wishart distribution  and the rescaled chi-square distribution respectively.  They approach  the normal distribution when the size of matrix is large enough. In Section~\ref{sec:results}, we provide numerical examples for the Sinai billiard and the Stadium billiard. The Wishart distribution  and the rescaled chi-square distribution are verified and more interesting behaviors are obtained as the system changes from nonchaotic to chaotic.   In Section~\ref{sec:discuss}, we conclude with a discussion of open questions.

\section{Preliminary: Krylov operator space}
\label{sec:algorithm}

Here we briefly review the algorithms of Krylov operator space~\cite{Parker:2018yvk}. Readers familiar with it can skip this section. For a quantum system with the Hamiltonian $H$, the time evolution of a given operator $\mathcal{O}$ in the Heisenberg picture is given by
\begin{align}
    \mathcal{O}(t)=e^{iHt}\mathcal{O}(0)e^{-iHt}
\end{align}
Using the Baker-Campbell-Hausdorff formula, it can be expanded as
\begin{align}
    \mathcal{O}(t)=\sum_{n=0}^\infty\frac{(it)^n}{n!}\mathcal{L}^n \mathcal{O}(0),
\end{align}
where $\mathcal{L} \coloneq [H, \cdot ]$ is the Liouvillian superoperator. 
This  sequence of operators $\mathcal{O}(0)$, $\mathcal{L}\mathcal{O}(0)$, $\mathcal{L}^2\mathcal{O}(0)$, $\cdots $ spans a space associated with the operator $\mathcal{O}(t)$. This space is the Krylov subspace generated by $\mathcal{L}$ (and hence $H$) and $\mathcal{O}(0)$, that is, %%【】到底是叫Krylov space还是subspace
\begin{align}
    \mathcal{K}_m(\mathcal{L},\mathcal{O}(0))\equiv \left \{  \mathcal{O}(0),\mathcal{L}\mathcal{O}(0),\mathcal{L}^2\mathcal{O}(0),\cdots ,\mathcal{L}^m\mathcal{O}(0)\right \}
\end{align}
An orthonormal basis, called the Krylov basis, can be constructed for $\mathcal{K}_m(\mathcal{L},\mathcal{O}(0))$. This construction  also generates a sequence of important nubmers called the Lanczos coefficients, because the algorithm used is the Lanczos algorithm( or similar algorithms). The orthogonality is guaranteed by the inner product  $(A|B)=\operatorname{tr}[A^\dagger B]$, which can be considered as a simplification of the Wightman inner product at infinite temperature. 
The  Lanczos algorithm can be briefly stated  as follows:
\begin{algorithmic}[1]
\Require  $\|\mathcal{O}\| = \sqrt{(\mathcal{O}|\mathcal{O})}$
\State $b_0\gets 0$
\State $\mathcal{O}_{-1}\gets 0$
\State $\mathcal{O}_0\gets \mathcal{O}(0)/\|\mathcal{O}(0)\|$
\State $n\gets 1$
\While{$b_n \neq 0$}
\State $\mathcal{A}_n \gets \mathcal{L}\mathcal{O}_{n-1}-b_{n-1}\mathcal{O}_{n-2}$
\State $b_n \gets \|\mathcal{A}_n\|$
\State $\mathcal{O}_n \gets  \mathcal{A}_n/b_n$
\State $n\gets n+1$
\EndWhile
\end{algorithmic}
% \begin{enumerate}
%     \item $b_0\equiv0 ,\quad\mathcal{O}_{-1}\equiv0$
%     \item $\mathcal{O}_0\equiv\mathcal{O}(0)/\|\mathcal{O}(0)\|$ where,  $\|\mathcal{O}\|\equiv\sqrt{(\mathcal{O}|\mathcal{O})}$
%     \item For $n\ge1\colon\mathcal{A}_n=\mathcal{L}\mathcal{O}_{n-1}-b_{n-1}\mathcal{O}_{n-2}$
%     \item Set $b_n=\|\mathcal{A}_n\|$
%     \item If $b_n=0$ stop; otherwise set $\mathcal{O}_n=\mathcal{A}_n/b_n$ and go to step 3.
% \end{enumerate}
The Krylov basis is given by $\left \{ \left | O_0  \right ) ,\left | O_1   \right )\cdots \left | O_{n_k-1}   \right ) \right \} $ and the dimension of Krylov operator space is $n_k=dim(\mathcal{K}_m)$. The sequence of positive numbers $\left \{ b_n \right \} $ is called the Lanczos coefficient and is essentially the matrix element of the Liouvillian operator
 \begin{align}
     L_{nm}\equiv (\mathcal{O}_n|\mathcal{L}|\mathcal{O}_m)=\begin{pmatrix}0&b_1&0&0&\cdots\\b_1&0&b_2&0&\cdots\\0&b_2&0&b_3&\cdots\\0&0&b_3&0&\ddots\\\vdots&\vdots&\vdots&\ddots&\ddots\end{pmatrix}
 \end{align}

The Lanczos coefficient $\left \{ b_n \right \} $ is essential for  the evolution of the operator $\mathcal{O}(t)$. To see this, firstly expand the operator $\mathcal{O}(t)$  in terms of the Krylov basis 
\begin{align}
\label{eq:expended k basis}
    \mathcal{O}(t)=\sum_{n=0}^{n_k-1}i^n\varphi_n(t)\mathcal{O}_n,\quad \varphi_n(t)\equiv (\mathcal{O}_n |\mathcal{O}(t))/i^n,
\end{align}
and then substitute it into the Heisenberg equation, which leads to the following equation
\begin{align}
\label{eq:dyn-eqn-coeff}
    \dot{\varphi}_n(t)=b_n\varphi_{n-1}(t)-b_{n+1}\varphi_{n+1}(t),
\end{align}
with the dot being the time derivative and the initial condition being $\varphi_n(0)=\delta_{n0}\|O\|$. 
This equation suggests that the spreading behavior of an operator in the Krylov space resembles the dynamics of wave functions on a 1d chain, called the Krylov chain~\cite{Parker:2018yvk}, where $\varphi_n(t)$ is  the wave function and $b_n$ is the hopping amplitude. This makes it possible to connect the operator complexity in the Krylov space  with condensed matter physics on spin chains. The tridiagonal form of the Liouvillian operator is another feature of the Krylov chain. 
On this Krylov chain, the time-dependent average position acts naturally as an indicator of the spreading behavior and is defined as the K-complexity
\begin{align}
C_K(t)=\sum_{n=0}^{n_k-1} n\left|\phi_n(t)\right|^2.
\end{align}
The K-complexity has been playing an important role in the analysis of quantum chaos, see the review~\cite{Nandy:2024htc} for related works. 

Obviously, results in the Krylov space will vary when the initial operator $\mathcal{O}$ changes. One way to deal with this dependence on initial operators is to do statistics by sampling random initial operators from different probability distributions. We will focus on this approach in the remaining of this paper.

\section{Statistical properties of Lanczos coefficients}
\label{sec:statistics-property}

The perspective of the Krylov chain~\eqref{eq:dyn-eqn-coeff} provides an efficient way to analyze the spreading dynamics of an operator by connecting to the physics of quantum many-body systems. Firstly in~\cite{Dymarsky:2019elm}, the delocalization on the chain is proposed as the reason for the chaotic behavior of an operator on the many-body systems. Then in~\cite{Rabinovici:2021qqt}, the localization on the Krylov chain is proposed as the reason for the nonchaotic behavior of an operator on the XXZ spin chain. 

Here, the Lanczos coefficients plays an an essential role because it is the hopping amplitude on the Krylov chain. The sequence of Lanczos coefficients is more erratic in the nonchaotic case and less erratic in the chaotic case. This behavior is  explored~\cite{Rabinovici:2021qqt,Trigueros:2021rwj} on the Krylov chain as a phenomenon of the Anderson localization. The energy eigenstate $\mid \omega)$ on the Krylov chain can be decomposed in the Krylov basis as 
\begin{align}
\mid \omega)=\sum_{n=0}^{n_k-1} \psi_n \mid \mathcal{O}_n)
\end{align}
There is a zero frequency state whose solution  is given by~\cite{Rabinovici:2021qqt,Trigueros:2021rwj}
\begin{align}
\frac{\psi_{2 n}}{\psi_0}=(-1)^n \prod_{i=1}^n \frac{b_{2 i-1}}{b_{2 i}}, \quad n=0, \ldots, \frac{n_k-1}{2}.
\end{align}
The extent of localization can be extracted from this zero frequency state and  the disordered cofficients $b_n$~\cite{fleishman1977fluctuations}. The localization length is inversely proportional to  the  variance of the Lanczos coefficients~\cite{Rabinovici:2021qqt},  defined by
\begin{align}
\label{eq:variance}
    \sigma^2 \coloneqq \text{Var} (x_i), \, x_i = \ln{\left| \frac{b_{2i-1}}{b_{2i}}\right| }.
\end{align}
Intuitively, the variance $\sigma^2$  measures the magnitude of the erratic behavior of $b_n$. A smaller variance $\sigma^2$ means less erratic behavior of Lanczos coefficients and this gives a larger localization length, so the Krylov chain is more delocalized and this happens when the system is more chaotic. 
Results in~\cite{Rabinovici:2021qqt} show that the variance $\sigma^2$ is indeed smaller in the chaotic case. 

Note that the sequence $x_i$ and the variance $\sigma^2$ of Lanczos coefficients are  computed in the Krylov space $\mathcal{K}_m(\mathcal{L},\mathcal{O})$ of an initial operator $\mathcal{O}$.  In this paper, we study statistics over random initial operators, in order to deal with this dependence on initial operators. Obviously, it is crucial to analyze the resulting statistics of the sequence $x_i$ and the variance  $\sigma^2$. In the following, the $\langle \cdots \rangle$ means the ensemble average over a sample of initial operators. %%%这一段超级难写

To be specific, We focus on the statistics of the more interesting correlation matrix~\footnote{In statistics, this is the scatter matrix to estimate the covariance matrix of a distribution. Since we are not doing pure mathematics here, we prefer to use the more physical term 'correlation matrix'~\cite{Eynard:2015aea}.}  $x_i x_j$, and  the resulting distribution $f_{\sigma^2}$ of the variance $\sigma^2$ across a sample. 
Naively, the resulting statistics of  $x_i x_j$ and $f_{\sigma^2}$ depend on the  distributions of the initial operator $\mathcal{O}$. However, we find that their statistics are the Wishart distribution~\cite{wishart1928generalised} and  the rescaled chi-square distribution respectively after a qualitatively analysis, which are independent of the distribution of initial operators. As the  matrix size of the operator increases, they approach the normal distribution. In the remaining of this section, we firstly explain the statistics in a purely intuitive manner. Then we explain the statistics in a phenomenological manner.

\subsection{Intuition on the normal distribution}
\label{sec:normal-dist}

Firstly we give an intuitive explanation of the reason why the resulting distributions are normal distributions when the size of matrix is large. This is related with the nature of the Krylov algorithm itself and can be understood using the central limit theorem of probability theory. 

Intuitively, the matrices involved in the Krylov algorithm depend only on the energy levels and the initial random matrix. This  can be seen by  the following lemma proved by mathematical induction. 
\begin{lemma}
\label{lemma1}
Matrix elements $\mathcal{O}^{(k)}_{mn}$ of Krylov basis $\mathcal{O}_{k}$ in the energy representation can be written as: 
\begin{align}
    \mathcal{O}^{(k)}_{mn}=C^{(k)}_{mn} (b_1,b_2,\cdots ,b_{k};E_{mn}) O_{mn}^{(0)}
\end{align}
where  $E_{mn} \coloneqq E_m-E_n$.
\end{lemma}

\begin{proof}
$k=0,1 $ are clearly true, because
\begin{align}
    &\mathcal{O}^{(0)}_{mn}=\mathcal{O}^{(0)}_{mn}\\
    &\mathcal{O}^{(1)}_{mn}=\frac{1}{b_1} \mathcal{A}_{mn}^{(1)} = \frac{1}{b_1} (E_{mn} \mathcal{O}^{(0)}_{mn}), b_1=\|E_{mn} \mathcal{O}^{(0)}_{mn}\|
\end{align}
Assume that $k-1, k-2$ holds, then $k$ also holds, because
\begin{align}
    \mathcal{A}_{mn}^{(k)} &=E_{mn} \mathcal{O}^{(k-1)}_{mn} -b_{k-1}\mathcal{O}^{(k-2)}_{mn}\\
    &=[E_{mn} C^{(k-1)}_{mn} (b_1,b_2,\cdots ,b_{k-1};E_{mn})-b_{n-1}C^{(k-2)}_{mn} (b_1,b_2,\cdots ,b_{k-2};E_{mn}) ]  O_{mn}^{(0)}\\
    &=\tilde{C}^{(k)}_{mn} (b_1,b_2,\cdots ,b_{k-1};E_{mn}) O_{mn}^{(0)}
\end{align}
and
\begin{align}
\label{eq:b_O_statistics}
    b_k &= \|\tilde{C}^{(k)}_{mn} (b_1,b_2,\cdots ,b_{k-1};E_{mn}) O_{mn}^{(0)}\|\\
    \mathcal{O}^{(k)}_{mn}&=\frac{1}{b_k} \mathcal{A}_{mn}^{(k)}=C^{(k)}_{mn} (b_1,b_2,\cdots ,b_{k};E_{mn}) O_{mn}^{(0)}.
\end{align}
Then by mathematical induction the lemma is proved.
% Therefore by mathematical induction, we prove that the Lanczos coefficient can be written as follows:
% \begin{align}
%     b_k^2= \|\mathcal{A}_k\|^2= \sum_{mn}\tilde{C}^{(k)2}_{mn} (b_1,b_2,\cdots ,b_{k-1};E_{mn})\left |O_{mn}^{(0)}\right | ^2.   
% \end{align}
\end{proof}
%补充两个证明：1.两个正态分布的乘积，当其中一个的均值远大于另一个时，服从正态分布。2.当正态分布的均值很大时，这个正态分布开根号在中心附近近似服从正态分布

% \begin{align} b_n^2 & = \left | \mathcal{A}_{mn}^{(k)} \right |^2 \\& = \sum_{mn} (E_{mn} \mathcal{O}^{(k-1)}_{mn})^2  -2b_{k-1}\sum _{mn}(E_{mn} \mathcal{O}^{(k-1)}_{mn}\mathcal{O}^{(k-2)}_{mn})^2+b_{k-1}^2\sum _{mn}(\mathcal{O}^{(k-2)}_{mn})^2\end{align}

From this lemma, we see that all  Lanczos coefficients $b_k= b_k(O_{mn}^{(0)};E_{mn})$  depend only on the energy levels and the initial random matrix $O_{mn}^{(0)}$. In practical computation, the energy levels are truncated to $\mathcal{N}_{max}$ and the random matrix $\mathcal{O}^{(0)}_{mn}$ is of size $\mathcal{N}_{max}\times\mathcal{N}_{max}$. If the value of $\mathcal{N}_{max}$ is large, there would be tremendous random numbers.  Intuitively, we would expect that the central limit theorem comes into play when we have large-enough random numbers, and the resulting distribution can be estimated by the normal distribution. On the other hand, we can intuitively expect that  most information of the system is preserved with a large enough $\mathcal{N}_{max}$. We know that a chaotic system resembles the behavior of RMT from level statistics. Then a more accurate approximation to the system means more randomness are preserved in this case. So things are random enough such that we can use the central limit theorem to expect the normal distribution.

  % $ \mathcal{O}^{(0)}_{mn}$ and $ \mathcal{O}^{(1)}_{mn}$ follow a normal distribution, for most $mn$, $b_k\gg \mathcal{O}_{mn}$, By Lemma2 $b_k \mathcal{O}_{mn}$ obey normal distribution. Each item in bn follows a different normal distribution, so when N is large enough, bn converges to a normal distribution. 

\subsection{Distribution of the average correlation matrix $\langle x_i x_j \rangle$}
\label{sec:correlatoin-length}

% \textbf{standard estimator of RMT}

Now we qualitatively analyze the statistics of the average correlation matrix $\langle x_i x_j \rangle$. Suppose  $K$ continuous Lanczos coefficients are  selected from the Krylov space of a given initial operator, with $K$ being the algorithmic steps.  We shall take the approximation that the sequence $\vec{X}=(x_1,\ldots,x_K)$ is a $K$-component random vector. This is easily understood in the Krylov chain, where by analog with the Anderson localization the Lanczos coefficients are disordered.  Without the Krylov chain, it is indeed harder to accept this. Naively looking at the Lanczos algorithm, the Lanczos coefficients seem to be already predetermined by the initial operator, like an initial value problem of ODE. However, we can intuitively get the randomness from Lemma~\ref{lemma1}. The initial operator is already a random matrix and all Lanczos coefficients are generated from a complicated pipeline of arithmetic of energy levels with the random matrix. And the energy levels themselves are closely related with RMT in the case of chaos. So we can still accept the assumption that $\vec{X}$ is a random vector even without the help of the Krylov chain. After all, studies in the literature  have shown the erratic behavior of Lanczos coefficients in various kinds of systems, and  this assumption indeed leads to a successful interpretation of the results in this paper. 

After accepting the randomness of $x_i$, we go one step further by assuming that the $K$-component random vector $\vec{X}=(x_1,\ldots,x_K)$ is from a multivariate normal distribution $N_K(0, \Sigma)$ with zero mean and a covariance matrix $\Sigma$ of size $K$. The zero mean is easily understood again in the Krylov chain by the Anderson localization, where the disordered $b_n$ imply that $x_i$ are fluctuating around zero. The property of being normal can be intuitively understood from the analysis in the previous subsection~\ref{sec:normal-dist}. Previous numerical results also show this property~\cite{Trigueros:2021rwj,Rabinovici:2022beu}. 

With this assumption, the average correlation matrix $\langle x_i x_j \rangle$ is exactly the common estimator used in RMT~\cite{Eynard:2015aea}. For a sample of $N$ random initial operators, we get $N$ random vectors that form a matrix of size $N\times K$
\begin{align}
X=\left(\vec{X}_1, \vec{X}_2, \cdots, \vec{X}_N\right)^{\mathrm{T}}, \quad \vec{X}_i=\left(x^{(i)}_1, x^{(i)}_2, \cdots, x^{(i)}_K\right), \quad x^{(i)}_j \in \mathbb{R},
\end{align}
where the $\vec{X}_i$ are independent, identically distributed random vectors from $N_K(0, \Sigma)$. 
Then the sample covariance matrix $M$ is a symmetric square matrix of size $K$ defined as
\begin{align}
M=X^{\mathrm{T}} X = \sum_{i=1}^{N} \vec{X}^{\mathrm{T}}_i \vec{X}_i \sim W_K(\Sigma, N),
\end{align}
which has the probability distribution  $W_K(\Sigma, N)$ known as the Wishart distribution. 
The empirical sample covariance matrix $M$ is the most commonly used estimator for the covariance matrix $\Sigma$. 
If $\Sigma$ is given, the random vector $\vec{X}_i$ has the distribution 
\begin{align}
N_K(0, \Sigma)\left(\vec{X}_i\right)=\frac{e^{-\frac{1}{2} \operatorname{Tr} \vec{X}_i \Sigma^{-1} \vec{X}_i^{\mathrm{T}}}}{(2 \pi)^{K / 2} \sqrt{\operatorname{det} \Sigma}} \prod_{j=1}^K d x^{(i)}_j
\end{align}
the Wishart distribution is defined as
\begin{align}
P(M) =\frac{1}{2^{N K/2}(\operatorname{det} \Sigma)^{N/2} \Gamma_K\left(\frac{N}{2}\right)}(\operatorname{det} M)^{\frac{N-K-1}{2}} e^{-\frac{1}{2} \operatorname{Tr} \Sigma^{-1} M}, 
\end{align}
where $\Gamma_K$ is a multivariate generalization of the Gamma function
\begin{align}
\Gamma_K\left(\frac{N}{2}\right)=\pi^{K(K-1) / 4} \prod_{j=1}^K \Gamma\left(\frac{N}{2}-\frac{j-1}{2}\right).
\end{align}

Note that our average correlation matrix $\langle x_i x_j \rangle$ is exactly the sample covariance matrix $M$, because $\langle x_i x_j \rangle=M_{ij}/N$. So it can indeed capture the statistics behind the ensemble average and we can use properties of the Wishart distribution to understand its behavior. In practical computation, $N$ is the sample size of random initial operators and $K$ is the number of selected Lanczos coefficients for each initial operator. 

\subsection{Distribution $f_{\sigma^2}$ of the variance $\sigma^2$}
\label{sec:variance-dist}

Now we have a qualitative understanding of the statistics of the average correlation matrix $\left \langle x_ix_j \right \rangle $ from the aspect of the Wishart distribution. What about the statistics of the resulting distribution $f_{\sigma^2}$ of the variance? We expect it to follow the same distribution as the average correlation matrix.  

For convenience,  suppose the variance $\sigma_{\mathcal{O}}^2$ is computed from $K=2p$ Lanczos coefficients for a given initial operator $\mathcal{O}$. By definition~\eqref{eq:variance}, we have
\begin{align}
\label{eq:sigma-twoterm}
\sigma_{\mathcal{O}}^2 &= \frac{1}{p} \sum_{i = 1}^{p}  x_i^2-\frac{1}{p^2} \bigg(\sum _{i=1}^{p}   x_i \bigg)^2 = \bigg(\frac{1}{p}-\frac{1}{p^2}\bigg) \sum_{i = 1}^{p}   x_i^2-\frac{1}{p^2} \sum _{i\neq j}^{p}  x_i x_j.
\end{align}
Qualitatively, we can expect that the second term with $i \neq j$ can be neglected compared with the first term. This is true when the set size $K$ is large. This condition is satisfied in all practical computation of the variance $\sigma^2$ and the second term is  smaller than the first term by at least one order of magnitude. So we obtain qualitatively
\begin{align}
\label{eq:scaled-chi-square}
\sigma_{\mathcal{O}}^2 &\sim  \bigg(\frac{1}{p}-\frac{1}{p^2}\bigg) \sum_{i = 1}^{p}   x_i^2.
\end{align}

This is the general result and it is hard to tell which distribution it satisfies because $x_i$ belongs to a vector  $\vec{X}$ from a multivariate normal distribution $N_K(0, \Sigma)$. 
But in the simple case when $x_i$ are independent random variables from a univariate normal distribution  $N(0, \sigma_x)$, we can scale $x_i$ by an overall factor to make its standard deviation being $1$ and the resulting distribution $f_{\sigma^2}$ of the variance~\eqref{eq:scaled-chi-square} is proportional to the chi-square distribution by an overall factor. We will provide numerical examples for this simplified case of  rescaled chi-square distribution.  

For the purpose of practical computation,  the amount of random numbers is determined by $\mathcal{N}_{max}$, $K$ and the sample size $N$ of initial operators. If they are large enough, we would expect the resulting distribution to be normal. Usually the choice of $N=5000$, $\mathcal{N}_{max}=100$ and $K=500$ turns out to be large enough for the resulting distribution  to be the normal distribution  and this choice in consistent with~\cite{Hashimoto:2023swv}.

\section{Numerical examples}
\label{sec:results}

As an example, we study the Sinai billiard and the Stadium billiard which are  typical two-dimensional systems of chaos. The Hamiltonian is defined by
\begin{align}
\label{eq:B-H}
    H=p_x^2+p_y^2+V(x,y)
\end{align}
with the potential 
\begin{align}
\label{eq:potential}
    V(x,y)=\begin{cases} 0 & (x,y) \in \Omega \\ \infty & \mathrm{else} \end{cases}.
\end{align}
For the Sinai billiard, the region $\Omega$ is obtained by cutting out a circle of radius $l$ from an equilateral  triangle of side length $L$, as shown in Figure~\ref{fig:B-region-sinai}. The system is integrable when $ l=0$ and chaotic when $0<l<L$. For the Stadium billiard, the region $\Omega$ is obtained by combining a rectangle and a quarter circle as shown in Figure~\ref{fig:B-region-Stadium}. For convenience, the area  $\Omega$ is kept at 1 and the parameter $a\equiv l/L$ is introduced, whose value determines the degree of chaos for the system.
\begin{figure}[!htbp]
        \centering
        \begin{subfigure}[b]{0.48\textwidth}
             \centering
             \includegraphics[width=1\textwidth,height=0.8\textwidth]{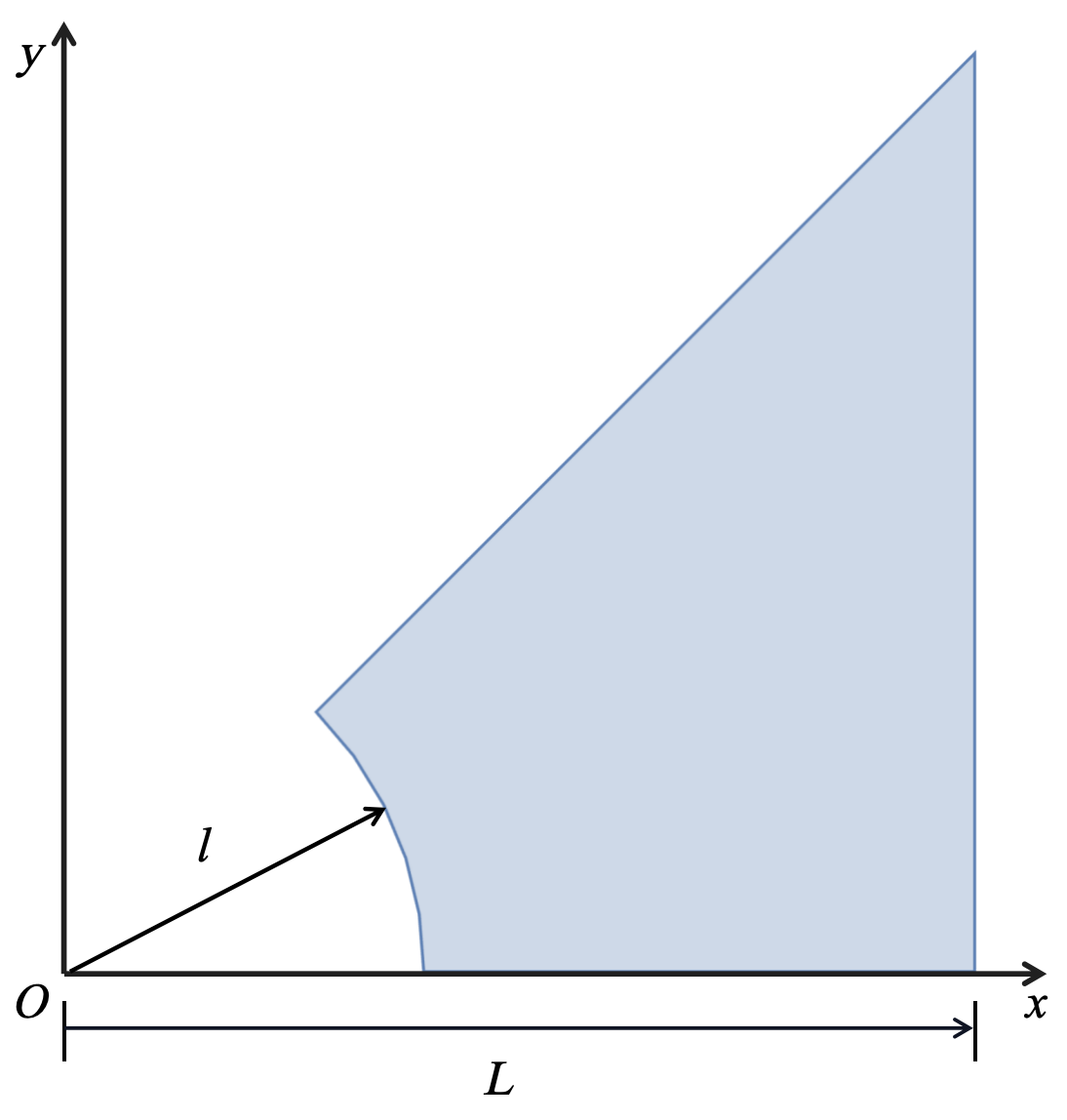} 
\caption{The Sinai billiard. Dirichlet conditions are imposed on the boundaries.}
\label{fig:B-region-sinai}
        \end{subfigure}
           \hfill
        \begin{subfigure}[b]{0.48\textwidth}
             \centering
             \includegraphics[width=1\linewidth,height=0.8\textwidth]{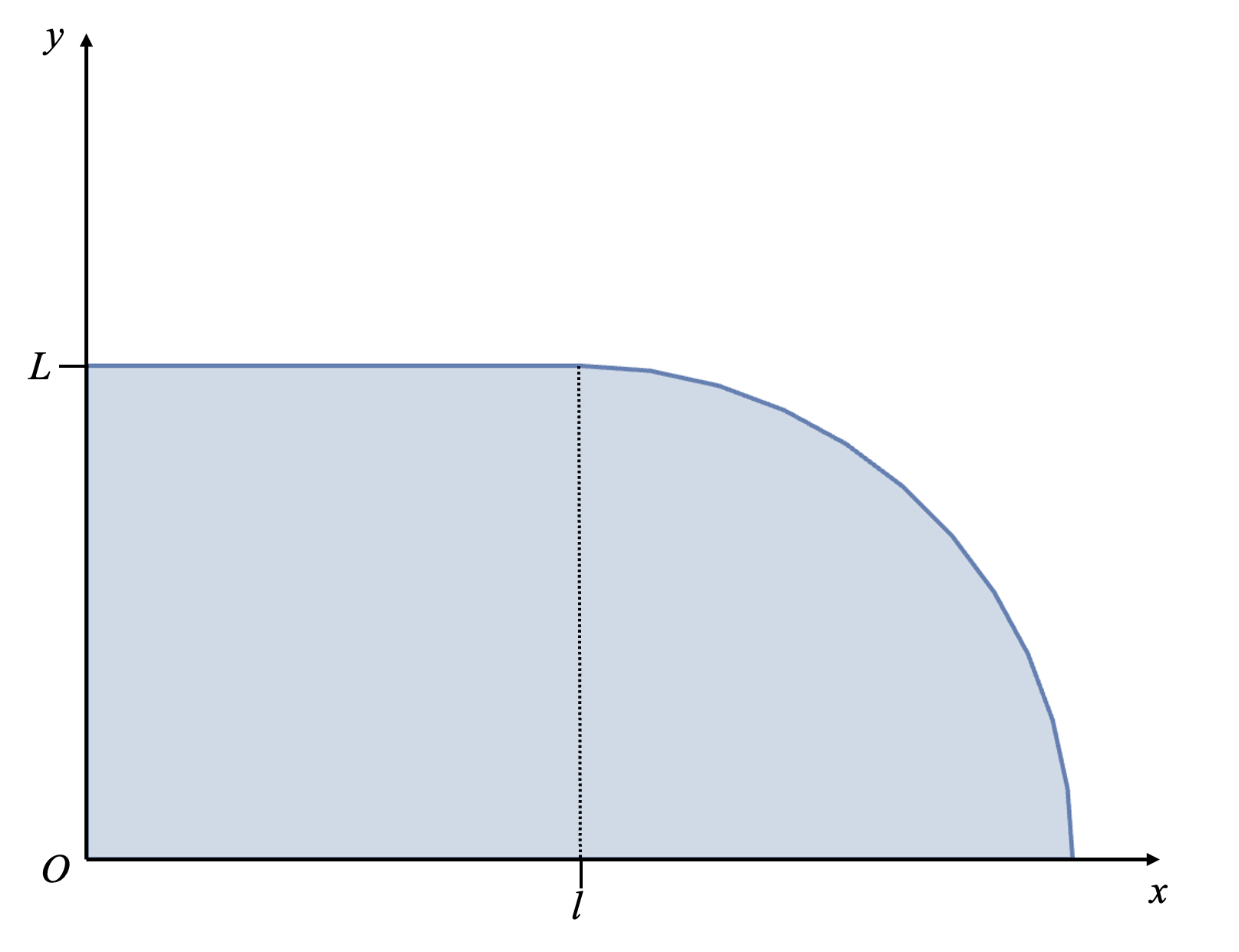} 
\caption{The Stadium billiard. Dirichlet conditions are imposed on the boundaries.}
             \label{fig:B-region-Stadium}
        \end{subfigure}     
        \caption{Region $\Omega$ of $V(x,y)$~\eqref{eq:potential} for the Sinai billiard and the Stadium billiard.}
\label{fig:B-region}
\end{figure}

For practical computation in the energy representation, we need to set the value of $\mathcal{N}_{max}$  to truncate the number of energy levels. It is adequate to take a value as long as the resulting distributions of  $\left \langle x_ix_j \right \rangle $ and  $f_{\sigma^2}$ are normal. When this happens, there are enough randomness in the statistics so that $\left \langle x_ix_j \right \rangle $ and  $f_{\sigma^2}$ can capture the relevant physics. Empirically, $\mathcal{N}_{max}=50$ is enough here. For the average correlation matrix  $\left \langle x_ix_j \right \rangle $, we use  $\mathcal{N}_{max}=50$ for the purpose of numerical efficiency. For the distribution $f_{\sigma^2}$ of the variance, we use  $\mathcal{N}_{max}=100$  to keep accordance  with~\cite{Hashimoto:2023swv}. For random initial operators, we choose several different probability distributions and compare their effects. Specifically, we consider random initial operators ($\mathcal{O}_{GOE}$, $\mathcal{O}_{GUE}$ and $\mathcal{O}_{uni}$) from the  Gaussian orthogonal ensemble (GOE), the Gaussian unitary ensemble (GUE)  and the uniform distribution.   In the case of GOE, the random matrix, $\mathcal{O}_{mn} \coloneqq \left\langle m \right| \mathcal{O}_{GOE}\left | n  \right \rangle  $, is a real symmetric matrix whose matrix elements satisfy the following joint probability density function~\cite{mehta1991random}:
\begin{align}
\label{eq:goe-prob}
    P\left ( \mathcal{O}_{11} ,\mathcal{O}_{12},\cdots,\mathcal{O}_{NN}   \right ) =(\frac{1}{2 \pi }) ^{\frac{N}{2}}(\frac{1}{ \pi }) ^{\frac{N^2-N}{2}}\exp \left [ -\frac{1}{2}\text{Tr} (\mathcal{O}_{GOE}^2)    \right ]. 
\end{align}
In the case of GUE,  the random matrix, $\mathcal{O}^{GUE}_{mn} \coloneqq  \left\langle m \right| \mathcal{O}_{GUE}\left | n  \right \rangle  $ is  Hermitian  whose elements satisfy the following joint probability density function:
\begin{align}
\label{eq:gue-prob}
&P\left ( \mathcal{O}^{(0)}_{11} ,\mathcal{O}^{(0)}_{12},\cdots,\mathcal{O}^{(0)}_{NN}   \right ) =(\frac{1}{2 \pi }) ^{\frac{N}{2}}(\frac{1}{ \pi }) ^{\frac{N^2-N}{2}}\exp \left [ -\frac{1}{2}\text{Tr} (\mathcal{O}^{(0)2})    \right ], \nonumber\\
&  P\left ( \mathcal{O}^{(1)}_{11} ,\mathcal{O}^{(1)}_{12},\cdots,\mathcal{O}^{(1)}_{NN}   \right ) =(\frac{1}{2 \pi }) ^{\frac{N}{2}}(\frac{1}{ \pi }) ^{\frac{N^2-N}{2}}\exp \left [ -\frac{1}{2}\text{Tr} (\mathcal{O}^{(1)2})    \right ], 
\end{align}
with $\mathcal{O}^{(0)}$ being the real part of $\mathcal{O}^{GUE}_{mn}$, and $\mathcal{O}^{(1)}$ being the imaginary part of $\mathcal{O}^{GUE}_{mn}$.
In the case of uniform distribution, three types of random operators are studied. Their random matrices belong to three types: purely real symmetric (URE), purely imaginary Hermitian (UIM), and  complex Hermitian (UCP). 
\begin{align}
\label{eq:uni-prob}
    \mathcal{O}_{mn}^{(0,1)} \sim U(-\sqrt{3} ,\sqrt{3}) 
\end{align}
% , whose elements satisfy the following joint probability density function
% \begin{align}
% \label{eq:uniform-prob}
%     P\left ( \mathcal{O}_{11} ,\mathcal{O}_{12},\cdots,\mathcal{O}_{NN}   \right ) =\left\{\begin{matrix} (\frac{1}{2})^{(N{^2+N})/2} &,\left [ -1,1 \right ] \\0 &,\text{else} \end{matrix}\right.
% \end{align}

Then we pick a random matrix $\mathcal{M}_{\mathcal{O}}$ of size $\mathcal{N}_{max} \times \mathcal{N}_{max}$ from the above probability distributions~\eqref{eq:goe-prob}-\eqref{eq:uni-prob}  for the initial random operator $\mathcal{O}$. After that, we compute the Lanczos coefficient $\left \{ b_n \right \} $ of $\mathcal{K}_m(\mathcal{L},\mathcal{M}_{\mathcal{O}})$ 
 by the algorithm in Section \ref{sec:algorithm}. To calculate the variance $\sigma^2$ of  Lanczos coefficients, we choose the set size to be $K=5\mathcal{N}_{max}$, in consistent  with~\cite{Hashimoto:2023swv}. Finally,   we repeat this process  $N=5000$ times to get the resulting distribution of $\left \langle x_ix_j \right \rangle $ and $f_{\sigma^2}$. To keep focus on the main physics, we put numerical details that are not essential to the main discussion into the Appendix.

\subsection{Distributions of $\langle x_i x_j \rangle$ and $\sigma^2$}

Firstly let's look at the Wishart distribution  $W_K(\Sigma, N)$ of the average correlation matrix $\langle x_i x_j \rangle$. The Wishart distribution can be viewed as a multi-variate generalization of the chi-squared distribution as follows: if $K=\Sigma=1$, it reduces to the chi-squared distribution with $N$ degrees of freedom. However, this is exactly the case in the numerical example:  the distributions of initial operators~\eqref{eq:goe-prob}-\eqref{eq:uni-prob} are univariate, whose $\Sigma$ is the identity matrix. So it is natural to expect that the resulting $K$-component random vector $\vec{X}$ is also univariate, that is, the  multivariate normal distribution $N_K(0, \Sigma)$ of $\vec{X}$  reduces to a univariate normal distribution  $N(0, \sigma_x)$. In this way, the Wishart distribution satisfied by $\langle x_i x_j \rangle$ reduces to the chi-square distribution of $\langle x_i^2 \rangle$. For a small matrix-size with $\mathcal{N}_{max}=15$ (that is away from the central limit), the histogram of $\langle x_i^2 \rangle$  is shown in  Figure~\ref{fig:B-chaosN5-bb_GUE} for the Sinai billiard and in  Figure~\ref{fig:B-chaosN5-bb_GUE-Stadium} for the Stadium billiard,  which  are indeed  the chi-square distribution. As the size of matrix increases, the distribution would become normal by the central limit theorem. For $\mathcal{N}_{max}=50$,  the corresponding histogram of  $\langle x_i x_j \rangle$ is plotted in  Figure~\ref{fig:B-chaosN50-bb_GUE} for the Sinai billiard and in  Figure~\ref{fig:B-chaosN50-bb_GUE-Stadium} for the Stadium billiard, and they are indeed the normal distribution. Here we only show the case when the system is chaotic and the initial distribution is GUE. For other cases there are similar properties, which are omitted because they do not provide extra insights and would ruin the conciseness and readability of the paper.   
\begin{figure}[!htbp]
        \centering
        \begin{subfigure}[b]{0.48\textwidth}
             \centering
             \includegraphics[width=1\linewidth]{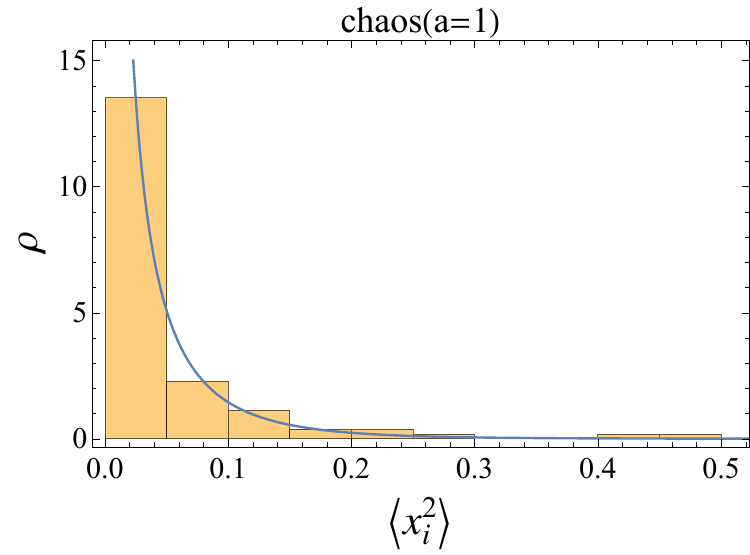} 
             \caption{Sinai billiard, $\langle x_i^2 \rangle$, $\mathcal{N}_{max}=15$.}
             \label{fig:B-chaosN5-bb_GUE}
        \end{subfigure}
           \hfill
        \begin{subfigure}[b]{0.48\textwidth}
             \centering
             \includegraphics[width=1\linewidth]{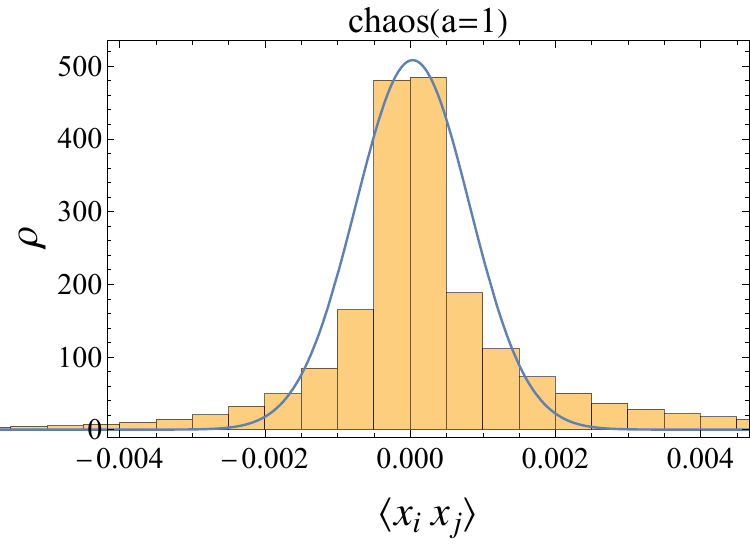} 
             \caption{Sinai billiard,  $\langle x_i x_j \rangle$, $\mathcal{N}_{max}=50$.}
             \label{fig:B-chaosN50-bb_GUE}
        \end{subfigure}    
        \begin{subfigure}[b]{0.48\textwidth}
             \centering
             \includegraphics[width=1\linewidth]{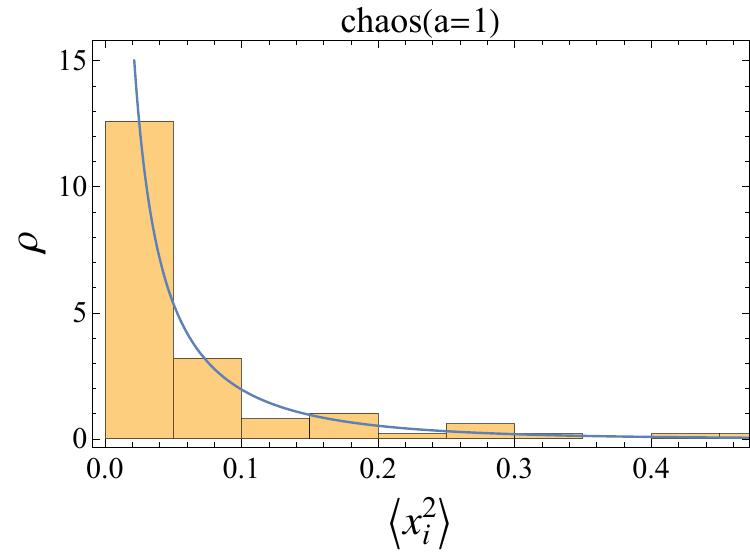} 
             \caption{Stadium billiard, $\langle x_i^2 \rangle$, $\mathcal{N}_{max}=15$.}
             \label{fig:B-chaosN5-bb_GUE-Stadium}
        \end{subfigure}
           \hfill
        \begin{subfigure}[b]{0.48\textwidth}
             \centering
             \includegraphics[width=1\linewidth]{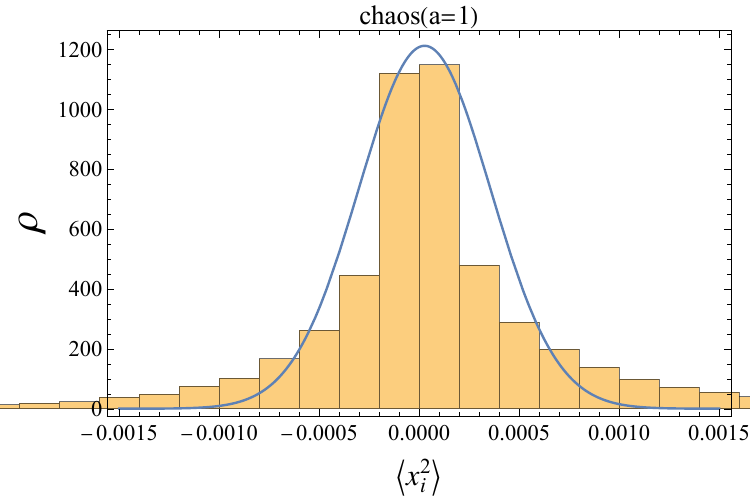} 
             \caption{Stadium billiard, $\langle x_i x_j \rangle$, $\mathcal{N}_{max}=50$.}
             \label{fig:B-chaosN50-bb_GUE-Stadium}
        \end{subfigure}             
        \caption{Wishart distribution of $\left \langle x_ix_j \right \rangle $ in the chaotic case for GUE.}
        \label{fig:B-bb-chi-chaos}
\end{figure}

Now let's look at the  resulting distribution $f_{\sigma^2}$ of the variance. In this numerical example  $x_i$ are essentially independent random variables from  $N(0, \sigma_x)$, so the resulting distribution $f_{\sigma^2}$ is the rescaled chi-square distribution  according to ~\eqref{eq:scaled-chi-square}.
Numerical results indeed support this expectation.
For a very small matrix-size with $\mathcal{N}_{max}=5$ (far away from the central limit), the histogram of the variance $\sigma^2$  is shown in  Figure~\ref{fig:chisquare-fsigma-GOE-rescale} for the Sinai billiard and in  Figure~\ref{fig:chisquare-fsigma-GOE-rescale-Stadium} for the Stadium billiard,  which  are indeed  the rescaled chi-square distribution with the curve being the fitted  distribution. As the size of matrix increases, the distribution would become normal by the central limit theorem. For $\mathcal{N}_{max}=100$,  the corresponding histogram of the variance $\sigma^2$  approaches the normal distribution and is plotted in  Figure~\ref{fig:chisquare-fsigma-GOE-normal} for the Sinai billiard and in  Figure~\ref{fig:chisquare-fsigma-GOE-normal-Stadium} for the Stadium billiard, where the curve are the fitted normal distribution. Here we only show for the case of GOE when the system is chaotic. For other cases there are similar results, which are omitted again for the conciseness and readability of the paper. 
\begin{figure}[!htbp]
        \centering
        \begin{subfigure}[b]{0.48\textwidth}
                \centering
                \includegraphics[width=\linewidth]{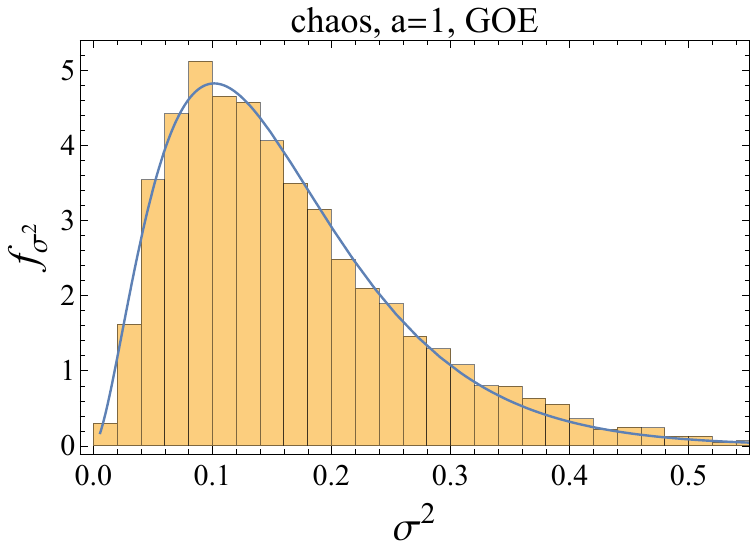}
                 \caption{Sinai billiard, $\mathcal{N}_{max}=5$, rescaled chi-square distribution.}
                 \label{fig:chisquare-fsigma-GOE-rescale}
        \end{subfigure}
           \hfill
        \begin{subfigure}[b]{0.48\textwidth}
                \centering
                \includegraphics[width=\linewidth]{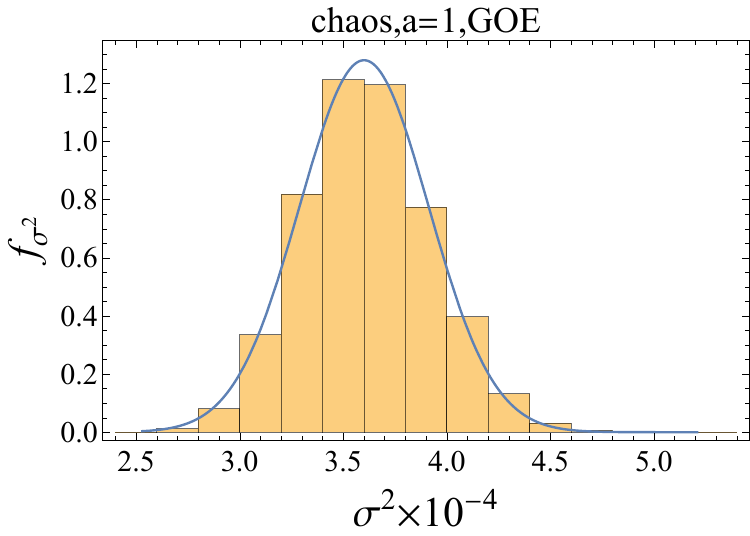}
                 \caption{Sinai billiard, $\mathcal{N}_{max}=100$, normal distribution. Skewness=0.201339, Kurtosis=3.19744}
                 \label{fig:chisquare-fsigma-GOE-normal}
        \end{subfigure}
        \begin{subfigure}[b]{0.48\textwidth}
                \centering
                \includegraphics[width=\linewidth]{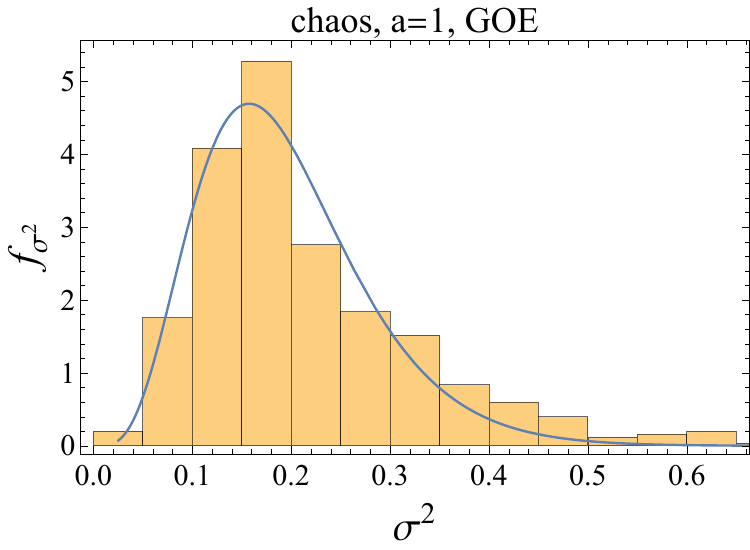}
                 \caption{Stadium billiard, $\mathcal{N}_{max}=5$, rescaled chi-square distribution.}
                 \label{fig:chisquare-fsigma-GOE-rescale-Stadium}
        \end{subfigure}
           \hfill
        \begin{subfigure}[b]{0.48\textwidth}
                \centering
                \includegraphics[width=\linewidth]{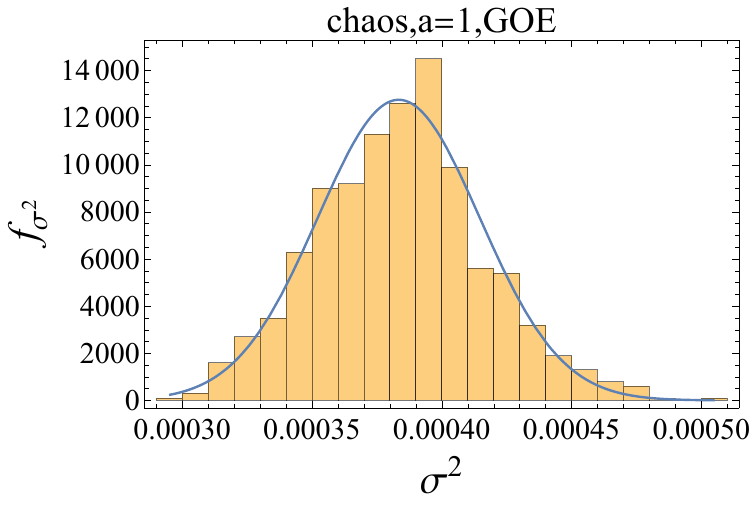}
                 \caption{Stadium billiard, $\mathcal{N}_{max}=100$, normal distribution. Skewness=0.18738, Kurtosis=3.05751.}
                 % 【】再强调一下这个分布不依赖选取的bn的数量？，nmax>15
                 \label{fig:chisquare-fsigma-GOE-normal-Stadium}
        \end{subfigure}
        \caption{Rescaled chi-square distribution of $f_{\sigma^2}$ in the chaotic case for GOE.}    
\end{figure}

\subsection{Behaviors between the nonchaotic and the chaotic case}

 \begin{figure}[h]
        \centering
        \begin{subfigure}[b]{0.48\textwidth}
             \centering
             \includegraphics[width=1\linewidth]{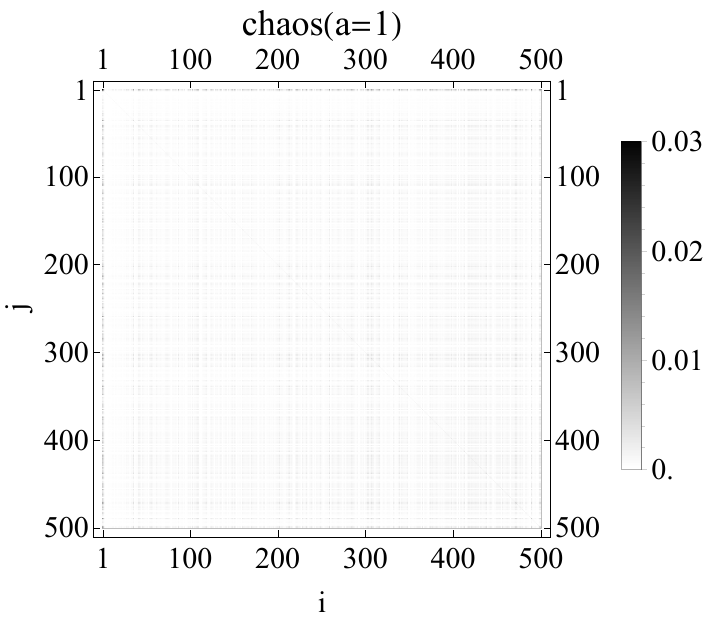} 
             \caption{Sinai billiard, $ \left | \langle  x_i x_j \rangle \right | $, $a=1$.}
             \label{fig:B-bnrelation_GUE_chaos}
        \end{subfigure}
           \hfill
        \begin{subfigure}[b]{0.48\textwidth}
             \centering
             \includegraphics[width=1\linewidth]{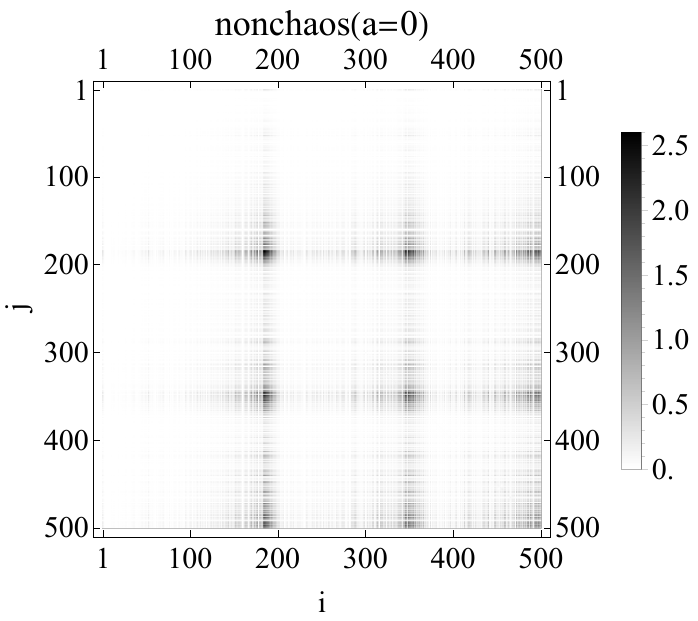} 
             \caption{Sinai billiard, $ \left | \langle  x_i x_j \rangle \right | $, $a=0$.}
             \label{fig:B-bnrelation_GUE_nonchaos}
        \end{subfigure}     
        \begin{subfigure}[b]{0.48\textwidth}
             \centering
             \includegraphics[width=1\linewidth]{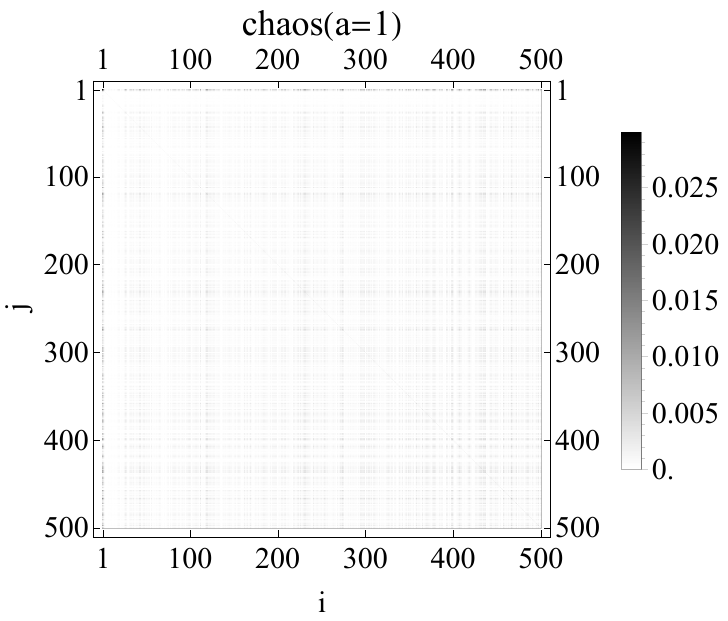} 
             \caption{Stadium billiard, $ \left | \langle  x_i x_j \rangle \right | $. $a=1$.}
             \label{fig:B-bnrelation_GUE_chaos-Stadium}
        \end{subfigure}
           \hfill
        \begin{subfigure}[b]{0.48\textwidth}
             \centering
             \includegraphics[width=1\linewidth]{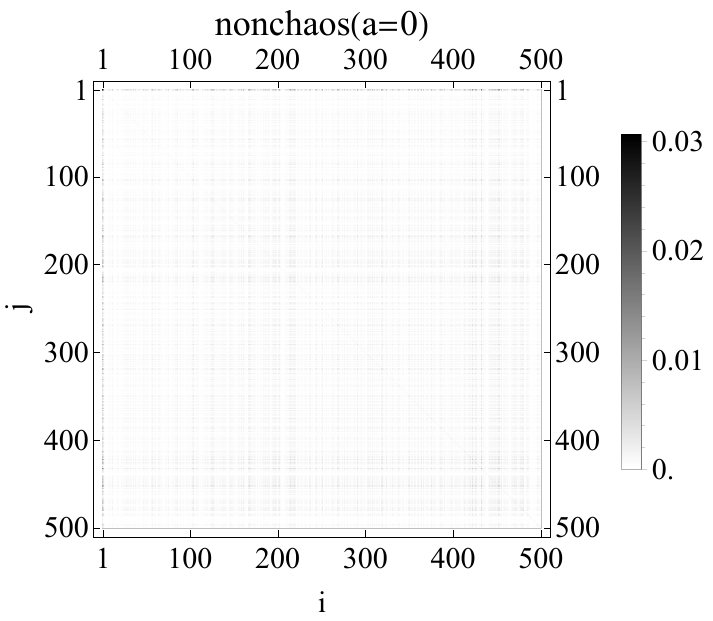} 
             \caption{Stadium billiard, $ \left | \langle  x_i x_j \rangle \right | $. $a=0$.}
             \label{fig:B-bnrelation_GUE_nonchaos-Stadium}
        \end{subfigure}             
        \caption{Patterns of $ \left | \langle  x_i x_j \rangle \right |  $ for the Sinai billiard and the Stadium billiard, $\mathcal{N}_{max}=50$.}
        \label{fig:B-bnrelation-50}
\end{figure}
 
\begin{figure}[h]
    \centering
    \begin{subfigure}[b]{1\textwidth}
        \centering
        \includegraphics[width=1\textwidth,height=0.45\textwidth]{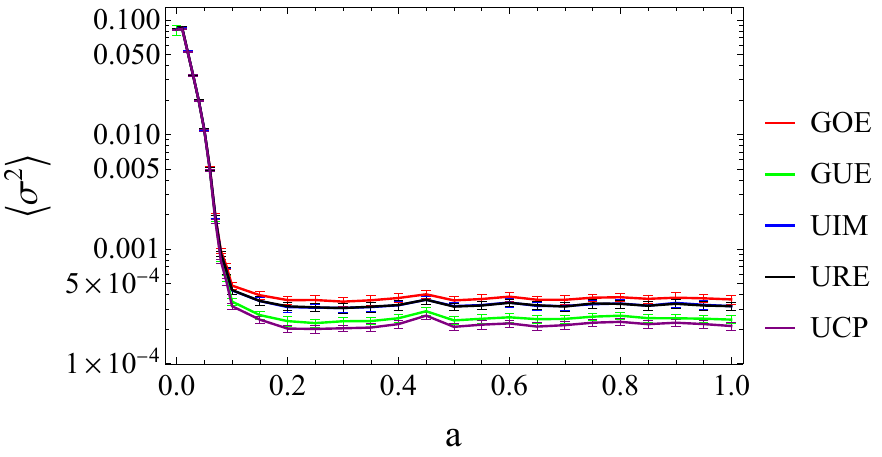} 
        \caption{Logarithm plot for the average  $\langle\sigma^2\rangle$ of $f_{\sigma^2}$. The error bar is the standard deviation. $\mathcal{N}_{max}=100$.}
        \label{fig:B-core-Sinai}
    \end{subfigure}
        \begin{subfigure}[b]{0.48\textwidth}
             \centering
             \includegraphics[width=1\linewidth]{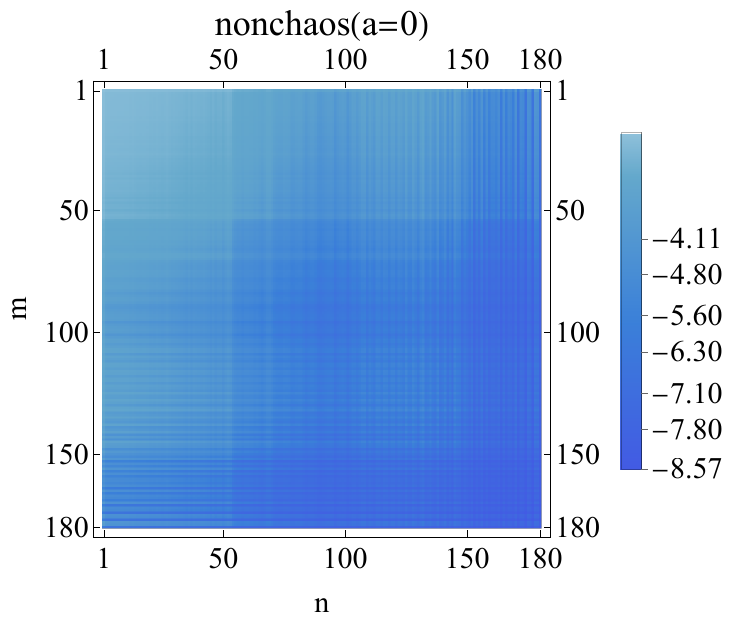}  
             \caption{$\left \langle \ln{|\psi_{2m}\psi_{2n}|} \right \rangle $, $a=0$, $\mathcal{N}_{max}=50$.}
             \label{fig:B-psirelation_GUE_chaos}
        \end{subfigure}
           \hfill
        \begin{subfigure}[b]{0.48\textwidth}
             \centering
             \includegraphics[width=1\linewidth]{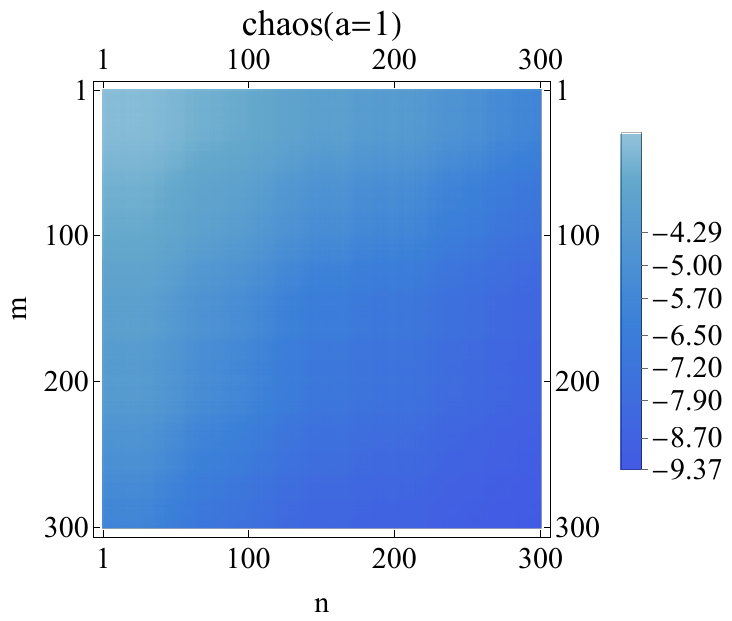}
             \caption{$\left \langle \ln{|\psi_{2m}\psi_{2n}|} \right \rangle $ , $a=1$, $\mathcal{N}_{max}=50$.}
             \label{fig:B-psirelation_GUE_nonchaos}
        \end{subfigure}      
    \caption{Sinai billiard: behaviors of $f_{\sigma^2}$ and $\left \langle \ln{|\psi_{2m}\psi_{2n}|} \right \rangle $ when the system changes from nonchaotic ($a=0$)  to chaotic ($a=1$).}
    \label{fig:different-behavior-Sinai}
\end{figure}

\begin{figure}[h]
    \centering
    \begin{subfigure}[b]{1\textwidth}
        \centering
        \includegraphics[width=1\textwidth,height=0.5\textwidth]{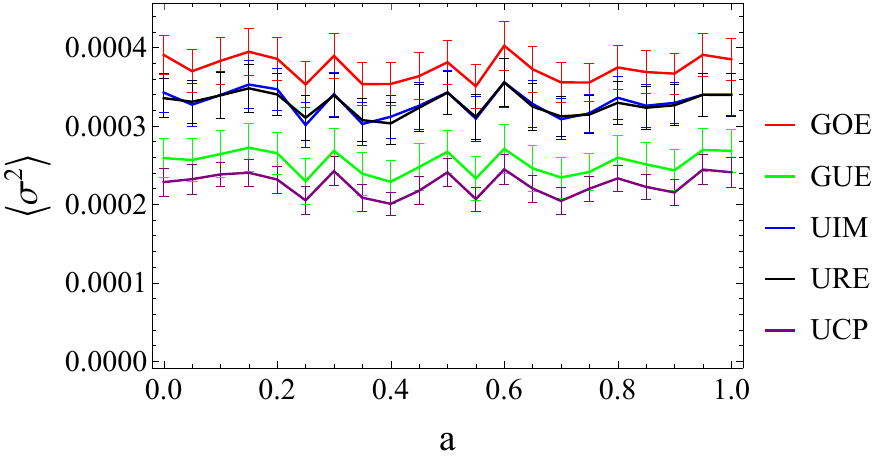} 
        \caption{Plot for the  average  $\langle\sigma^2\rangle$ of $f_{\sigma^2}$. The error bar is the standard deviation. $\mathcal{N}_{max}=100$. }
        \label{fig:st-core-Stadium}
    \end{subfigure}  
        \begin{subfigure}[b]{0.48\textwidth}
             \centering
             \includegraphics[width=1\linewidth]{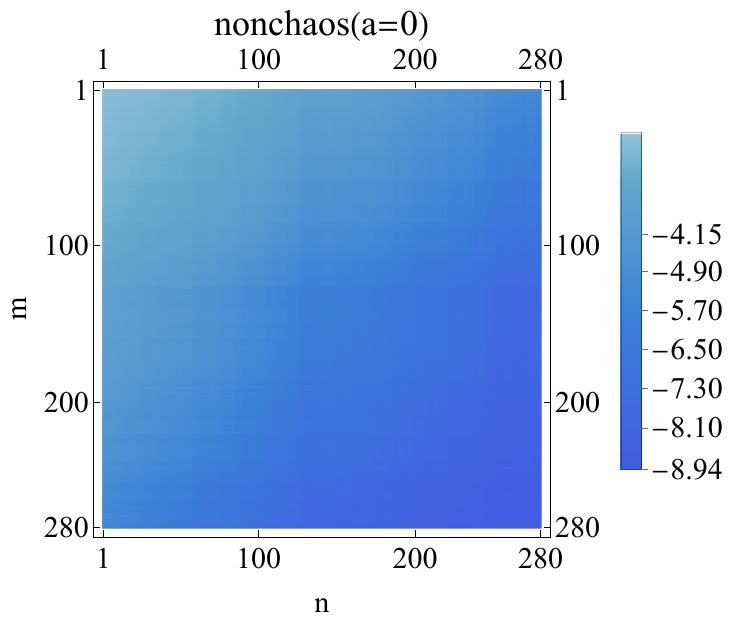}  
             \caption{$\left \langle \ln{|\psi_{2m}\psi_{2n}|} \right \rangle $, $a=0$, $\mathcal{N}_{max}=50$.}
             \label{fig:B-psirelation_GUE_chaos-Stadium}
        \end{subfigure}
           \hfill
        \begin{subfigure}[b]{0.48\textwidth}
             \centering
             \includegraphics[width=1\linewidth]{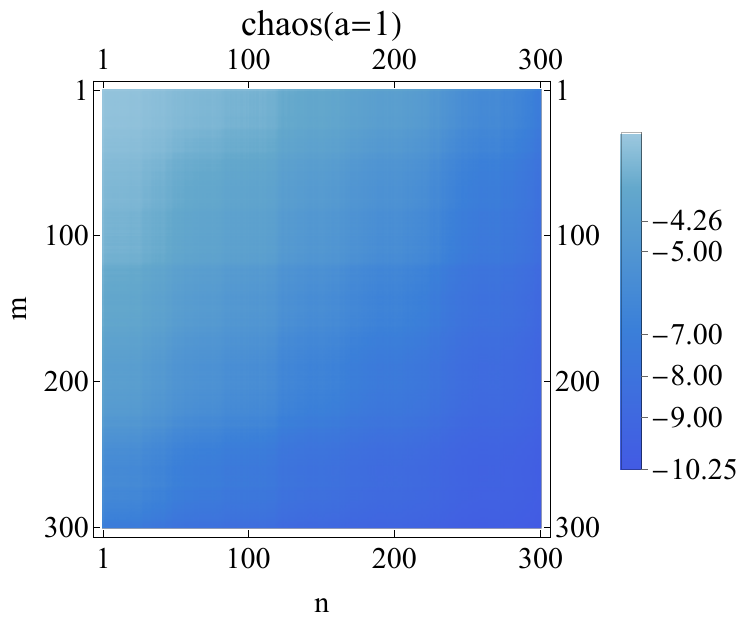}
             \caption{$\left \langle \ln{|\psi_{2m}\psi_{2n}|} \right \rangle $ , $a=1$, $\mathcal{N}_{max}=50$.}
             \label{fig:B-psirelation_GUE_nonchaos-Stadium}
        \end{subfigure}
    \caption{Stadium billiard: behaviors of $f_{\sigma^2}$ and $\left \langle \ln{|\psi_{2m}\psi_{2n}|} \right \rangle $ when the system changes from nonchaotic ($a=0$)  to chaotic ($a=1$).}
    \label{fig:different-behavior-Stadium}
\end{figure}

\textbf{[I] Correlation matrix:} For the Sinai billiard,
the average correlation matrix $\left \langle x_ix_j \right \rangle$  has different patterns between the nonchaotic and the chaotic case, shown in Figure~\ref{fig:B-bnrelation_GUE_chaos} and ~\ref{fig:B-bnrelation_GUE_nonchaos}  for the case of GUE. 
With $\mathcal{N}_{max}=50$, its  histogram approaches the normal distribution as shown in Figure~\ref{fig:B-chaosN50-bb_GUE} so that there is enough randomness in the statistics. There is a clear distinction in the pattern between the chaotic and the nonchaotic case. In the nonchaotic case, there are heavy cross-like structures with large matrix values. But in the chaotic case, these heavy cross-like structure disappears (or fades away) and matrix values  are more uniformly or smoothly distributed. 

For the Stadium billiard, the average correlation matrix $\left \langle x_ix_j \right \rangle$ has no clear distinction  between the chaotic and the nonchaotic case, as shown in Figure~\ref{fig:B-bnrelation_GUE_chaos-Stadium} and ~\ref{fig:B-bnrelation_GUE_nonchaos-Stadium}  for the case of GUE. With $\mathcal{N}_{max}=50$, its  histogram approaches the normal distribution as shown in Figure~\ref{fig:B-chaosN50-bb_GUE-Stadium} so that there is enough randomness in the statistics. However, matrix values are smoothly distributed without cross-like structures in both the chaotic and the nonchaotic case. 

From the aspect of RMT, the average correlation matrix $\left \langle x_ix_j \right \rangle$ is  the common estimator of the statistics. The similarity of Figure~\ref{fig:B-bnrelation_GUE_chaos}, ~\ref{fig:B-bnrelation_GUE_chaos-Stadium} and ~\ref{fig:B-bnrelation_GUE_nonchaos-Stadium} tells us that they have similar statistics. This is interesting because the Stadium billiard in all phases, whether chaotic or not, are close to the nonchaotic Sinai billiard. 

On the other hand, $x_i$ characterizes the disorderedness of Lanczos coefficients from the aspect of the Krylov chain. So $\left \langle x_ix_j \right \rangle$ quantitatively captures the erratic behavior and the correlation of Lanczos coefficients. The large matrix values for the nonchaotic Sinai billiard (Figure~\ref{fig:B-bnrelation_GUE_nonchaos}) indicate more erratic behavior of Lanczos coefficients and more localization on the Krylov chain. The Stadium billiard (Figure~\ref{fig:B-bnrelation_GUE_chaos-Stadium}, ~\ref{fig:B-bnrelation_GUE_nonchaos-Stadium}) and the chaotic Sinai billiard (Figure~\ref{fig:B-bnrelation_GUE_chaos}) mean that less erratic behavior there.

Furthermore, the repetition of the heavy cross-like structure in the nonchaotic Sinai billiard (Figure~\ref{fig:B-bnrelation_GUE_nonchaos}) is  a very interesting phenomenon. It might come from some abstract periodic boundary conditions  of the Krylov chain. Or it might be connected to the phenomenon of two-electron correlations of Anderson localization on the Krylov chain~\cite{Weinmann1995h2eOF}. To make this connection with Anderson localization, we plot the average correlation matrix of the wave function $\langle\ln{|\psi_{2m}\psi_{2n}|}\rangle$ for the zero frequency state of the Krylov chain. For the nonchaotic Sinai billiard,  there are also heavy cross-like structures for the correlation matrix of the wave function as shown in Figure~\ref{fig:B-psirelation_GUE_nonchaos}. These cross-like structures locally look like the cross-like structure of the  two-electron correlations of Anderson localization in disordered mesoscopic ring~\cite{Weinmann1995h2eOF}. And the cross-like structure of $\ln{|\psi_{2m}\psi_{2n}|}$ disappears in the chaotic Sinai billiard in Figure~\ref{fig:B-psirelation_GUE_chaos}. For the Stadium billiard, the  correlation matrix of the wave function is shown in Figure~\ref{fig:B-psirelation_GUE_chaos-Stadium} for the chaotic case and in Figure~\ref{fig:B-psirelation_GUE_nonchaos-Stadium} for the nonchaotic case. They are nearly the same and is also similar to the chaotic Sinai billiard (Figure~\ref{fig:B-psirelation_GUE_chaos}), without cross-like structure. This behavior is consistent with the 
statistics  $\left \langle x_ix_j \right \rangle$.

\textbf{[I] distribution of the variance $\sigma^2$:} 
For the Sinai billiard, its scatter plots  in Figures~\ref{fig:billiard-sigmaG-behavior-100} and ~\ref{fig:billiard-Histdistribution-100} demonstrate that the resulting distributions $f_{\sigma^2}$ also have different behaviors between  chaotic and nonchaotic cases. 
Their axis values indicate that they can be compared together in one plot.
So we fit~\footnote{ The details of this data fitting is shown in the Appendix.} their histogram to the normal distribution and find the average value $\langle \sigma^2 \rangle$ and the standard deviation, which can characterize the resulting distribution $f_{\sigma^2}$.  Then the resulting distribution $f_{\sigma^2}$ can be shown in one plot where the  point is this average value and the error bar is the standard deviation. In Figure~\ref{fig:B-core-Sinai}, the resulting distributions  $f_{\sigma^2}$ are shown for all different distributions (GOE, GUE, URE, UIM and UCP) of initial operators, as the system varies from nonchaotic $a=0$ to chaotic $a=1$. When the system is nonchaotic, the data points of all distributions overlap together within the range of error bars. As the system becomes chaotic, different distributions start to separate approximately at $a=0.1$ and the distance of separation is roughly fixed after $a=0.15$. Note that the axis is logarithmic, so the separation is very large. We see that they split into two groups well-separated: one group is of GOE, URE and UIM, the other group is GUE and UCP.  

For the Stadium billiard,  the resulting distributions  $f_{\sigma^2}$ always split into two  well-separated groups  as shown in Figure~\ref{fig:st-core-Stadium}, whether the system is chaotic or not. Note that this is not a logarithmic plot because the values are nearly constants. 
Interestingly, the two groups match those observed in the case of the Sinai billiard of Figure~~\ref{fig:B-core-Sinai}. Comparing the values in these two plots, we see that the average $\langle\sigma^2\rangle$ is nearly identical: in one group  it remains approximately $0.00035$ while in the other it stays around  $0.00025$. From these values of the resulting distributions  $f_{\sigma^2}$, the Stadium billiard exhibits behavior similar to the chaotic Sinai billiard.

\section{Conclusion}
\label{sec:discuss}

In this paper, we studied the statistics involved by Lanczos coefficients in the Krylov operator complexity and deeply exhibited the connection among RMT, Anderson localization and Krylov complexity.  We focused on  the average correlation matrix $\langle x_{i} x_{j}\rangle$  and the resulting distribution  $f_{\sigma^2}$ of the variance. We found in a phenomenological way that they satisfy  the Wishart distribution and the (rescaled) chi-square distribution respectively, independent of the distributions of initial operators. With the help of these two quantities, we explored the Sinai and the Stadium billiard as a numerical example. The statistical distributions are verified and further interesting behaviors are obtained as the system changes from nonchaotic to chaotic, which  indicates a consistent connection among RMT, Anderson localization and Krylov complexity. 

For the Sinai billiard, the nonchaotic case is clearly distinguished from the chaotic case. The nonchaotic case has more erratic behavior as captured by  $\left \langle x_ix_j \right \rangle$ which is consistent with the analysis of~\cite{Dymarsky:2019elm,Rabinovici:2021qqt}, and has heavy cross-like structure as capture by $\ln{|\psi_{2m}\psi_{2n}|}$  which locally looks like two-electron correlations of Anderson localization~\cite{Weinmann1995h2eOF}. The resulting distribution $f_{\sigma^2}$ from different ensembles are overlapped with each other. 
However, for the Stadium billiard, there is no distinction between the nonchaotic and the chaotic case. More interestingly, the whole behavior of the Stadium billiard is similar to the chaotic Sinai billiard. Intuitively we can see that they have similar behavior in the chaotic case. But why the nonchaotic Stadium billiard is also similar to the chaotic case? The consistency among the three quantities $\left \langle x_ix_j \right \rangle$, $\ln{|\psi_{2m}\psi_{2n}|}$ and $f_{\sigma^2}$ indicates this result is reliable. So there should be a reason for it. Currently this is an open question. 

Another interesting behavior is the two groups of separation for the resulting distribution $f_{\sigma^2}$. The GOE, URE and UIM are in one group. We can view  GOE, URE and UIM as essentially composed of only real numbers, because the UIM is just real numbers multiplied by the overall factor-the imaginary $i$. The other group is GUE and UCP, which are essentially composed of complex numbers with both real and imaginary parts. This separation seems to be connected with the dimension of random numbers, which might be  related to the nature of the integrability-breaking-term  or to the symmetry of the model~\cite{haake1991quantum}. This remains to be one of the open problems. 

In addition, we emphasize the importance of the resulting distribution being the normal distribution.  We repeat the computation for $\mathcal{N}_{max}=5$ in Appendix, where the resulting distribution $f_{\sigma^2}$ is the rescaled chi square and not normal. Then there is no separation of $f_{\sigma^2}$ in all cases. Note that this rescaled chi square behavior of $\mathcal{N}_{max}=5$ further proves the phenomenological analysis in Section~\ref{sec:statistics-property}. Because few energy levels are kept with $\mathcal{N}_{max}=5$ it is impossible to extract the statistics numerically in a bottom-up manner. However, once the statistics is given by phenomenological analysis, the numerical results of $\mathcal{N}_{max}=5$ can prove the statistics in a top-down manner. 

Qualitatively the resulting distribution being normal means that we have enough randomness in the sample, so that the statistics is reliable and robust. Empirically, a choice of $\mathcal{N}_{max}=50$ is already adequate to get the normal distribution. It remains an open problem to quantitatively understand the connection between being normal distributions and the separation behavior of $f_{\sigma^2}$.

Finally, we explain a connection (or equivalence) between the average over initial operators and the RMT Hamiltonians. For the chaotic case,  a similar behavior-of-separation is obtained related with RMT~\cite{Rabinovici:2022beu}. For a spin system~\cite{Rabinovici:2022beu} under chaotic dynamics, the Hamiltonian is sampled from RMT ensembles while the initial operator is fixed. For Hamiltonians from GUE and GOE,  a separation is obtained in the late-time saturation of the Krylov complexity. Qualitatively, the behavior of late-time saturation of K-complexity are consistent with the behavior of variances $\sigma^2$ of the Lanczos coefficients~\cite{Rabinovici:2021qqt}. So it might be possible that the ensemble average of the variance $\sigma^2$ over RMT Hamiltonian are also separating in~\cite{Rabinovici:2022beu}, although it is not computed there. If this indeed happen, the separation behavior of the resulting distributions $f_{\sigma^2}$ and of the late-time saturation of the Krylov complexity can be viewed as  'dual' to each other for chaotic systems.  

Another sign for this 'duality' also appears in the study of the spin chain~\cite{Trigueros:2021rwj}. There the randomness is introduced by the random transverse field from a uniform distribution. So it is equivalent to a RMT Hamiltonian. There the statistics is focused on the parameters of the linear extrapolation formula of Lanczos coefficients. These fitted parameters also follow the normal distribution, similar to the resulting distribution $f_{\sigma^2}$, although they are completely different quantities. This suggests that the statistics would be normal as long as there is enough randomness, whether it is due to randomn initial operators or RMT Hamiltonians. 

Qualitatively, this 'duality' can be understood as follows: in one side, the    Hamiltonian are random and the initial operator is fixed, and in the other side, the Hamiltonian is fixed and the initial operators are random. They are 'equivalent' when taking the statistics, at least for chaotic systems. From the proof~\eqref{eq:b_O_statistics} of the lemma in Section~\ref{sec:statistics-property}, we see that the Lanczos coefficients and the Krylov basis depend only on the initial operators and energy levels. Both the random energy levels and random initial operators can give random Lanczos coefficients, so intuitively, these statistics can be viewed as dual or equivalent.  It remains an open problem to quantitatively prove this duality. In~\cite{Rabinovici:2021qqt}, the Lanczos coefficients are rewritten in terms of the Hankel determinants of level spacings, which might be helpful in this direction.

\acknowledgments
Wei Fan is supported in part by the National Natural Science Foundation of China under Grant No.\ 12105121.

\appendix

\section{Numerical details not relevant to the main physics}

Here we only show the numerical details for the Sinai billiard to keep the conciseness of the paper. For the Stadium billiard, most results are parallel except that there is no overlapping behavior in the nonchaotic case of $\mathcal{N}_{max}=100$. So there is no overlapping behavior in its plots analogous to Part \textbf{[II]}  below, which has already be quantized in Figure~\ref{fig:st-core-Stadium} in the main context. 

\textbf{[I]:} Firstly an example of Lanczos coefficients $b_n$  is shown in  Figure~\ref{fig:B-bn-example-100},  for the case of GUE. There is a clear difference between the chaotic and the nonchaotic case, which is explored further by the average correlation matrix $\left \langle x_ix_j \right \rangle $. The set of coefficients $\left \{ b_n \right \} $ between the two red lines, $5\mathcal{N}_{max}\le n\le 10\mathcal{N}_{max}$, is selected to compute the variance $\sigma^2$.  Note that this selection of the set  does not qualitatively affect the result of this paper. The effect of selecting the set of  $\left \{ b_n \right \} $ from different algorithm step sizes is discussed in the end of this Appendix.

\textbf{[II]:} Then the variances $\sigma^2$  are shown in Figure~\ref{fig:billiard-sigmaG-behavior-100} for the samples of different distributions (GOE, GUE, URE, UIM and UCP) of initial operators, where $m$ labels the sample. Note that only the first 1000 out of the total 5000 samples are displayed for visual clarity, as the remaining samples exhibit similar behavior.   For GOE and GUE, they overlap completely in the nonchaotic case in Figure~\ref{fig:B-Gaussian-100-nonchaos}, and separate from each other in the chaotic case in Figure~\ref{fig:B-Gaussian-100-chaos}. Similary for the uniform distributions,  they mix together in the nonchaotic case in Figure~\ref{fig:B-uni-sigma-100-nonchaos}, and split into two separated groups in the chaotic case in Figure~\ref{fig:B-uni-sigma-100-chaos}, with URE and UIM being one group and UCP being the other.  

The histogram of all the sampled $\sigma^2$ is shown in Figure~\ref{fig:billiard-Histdistribution-100} and the resulting probability distribution resembles the  normal distribution. For the Gaussian ensembles,  we can see clearly that they are overlapped in the nonchaotic case  in Figure~\ref{fig:B-GaussianHist-100-nonchaos} and separated in the chaotic case  in Figure~\ref{fig:B-GaussianHist-100-chaos}. For the uniform distributions, they are overlapped in the nonchaotic case in Figure~\ref{fig:B-uniformHist-100-nonchaos}. In the  chaotic case in Figure~\ref{fig:B-uniformHist-100-chaos}, URE and UIM are still together, but they are separated from UCP. 

Now we show the data fitting  of the histogram in Figure~\ref{fig:billiard-Histdistribution-100} to the normal distribution. For GOE and GUE, the fitted distribution is shown in Figure~\ref{fig:billiard-fitted-distribution-Gaussian-100}. For uniform distributions, the fitted distribution is shown in Figure~\ref{fig:billiard-fitted-distribution-uniform-100}. The average value and the standard deviation are listed in Table~\ref{tab:billiard-fitted-data-100}, where we compare the values obtained from the sampled data and from the fitted normal distribution. 

Note that the corresponding results of the Stadium billiard are always separating whether the system is chaotic or not, that it they are similar to   Figure~\ref{fig:B-Gaussian-100-chaos} and Figure~\ref{fig:B-GaussianHist-100-chaos}. 

\textbf{[III]:} After that, let's verify the chi-square distribution of $\sigma^2$ for the case $\mathcal{N}_{max}=5$. 
The sets of coefficients $b_n$ are chosen from step-sizes $1<n<15$.  We repeat the computation  and show that there is neither normal distributions nor separating behaviors. 

The histogram of all $5000$ samples are shown in Figure~\ref{fig:billiard-hist-distribution-GE-5}. Obviously they are not normal distributions and are not separated in the  chaotic case. So the normal distribution is necessary for the separating behavior. 

We fit the histogram to the rescaled chi-square distribution, the results are shown in Figure~\ref{fig:billiard-his-distribution-fit-GE-5} for the Gaussian ensembles and in Figure~\ref{fig:billiard-fitted-distribution-uniform-5} for the uniform distributions. We see that the  rescaled chi-square distribution indeed captures the resulting distribution $f_{\sigma^2}$, so it verifies our analysis in Section~\ref{sec:variance-dist}. 

Note that the corresponding results of the Stadium billiard are the same here, that it they are overlapping and satisfy the rescaled chi square distribution.

\textbf{[IV]:} Finally we show the independence on the sets of coefficients used in computing $\sigma^2$. Let's choose the Lanczos coefficients  $\left \{ b_n \right \} $ from the new step-sizes $10\mathcal{N}_{max}\le n\le 15\mathcal{N}_{max}$ of Figure~\ref{fig:B-bn-example-100}. Firstly, we compute the variance $\sigma^2$~\eqref{eq:variance}, where the first $1000$ samples are shown in Figure~\ref{fig:billiard-sigmaG-choice-100}. We see that they mix together in the nonchaotic case and separate from each other in the chaotic case. Then, the histogram of all $5000$ samples are shown in Figure~\ref{fig:billiard-Histdistribution-choice-100}. Obviously they  still resemble normal distributions and are overlapped in the nonchaotic case and well-separated in the  chaotic case. Finally, the fitted normal distributions are shown in Figure~\ref{fig:billiard-fitted-distribution-Gaussian-choice-100} and~\ref{fig:billiard-fitted-distribution-uniform-choice-100}. So we still get the resulting normal distributions and the overlapping-splitting behaviors, although the detailed values of the average and the standard deviation change.

\begin{figure}[h]
        \centering
        \begin{subfigure}[b]{0.48\textwidth}
             \centering
             \includegraphics[width=1\linewidth]{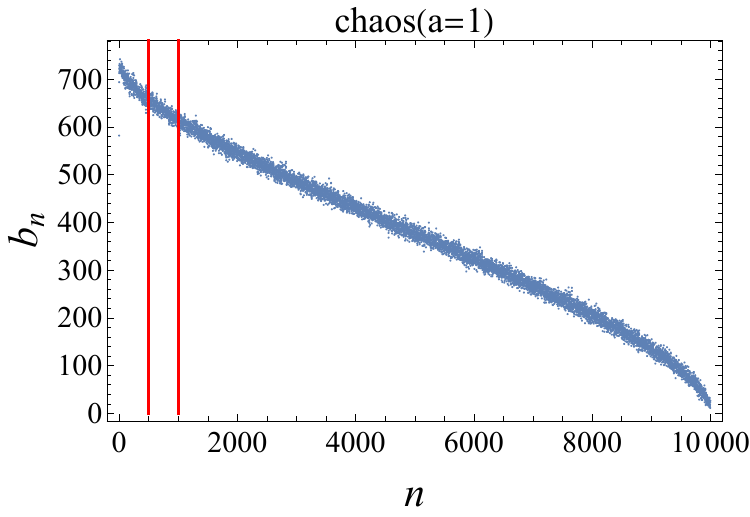} 
             \caption{$b_n$ in chaotic case $a=1$.}
             \label{fig:B-bn_GUE_chaos}
        \end{subfigure}
           \hfill
        \begin{subfigure}[b]{0.48\textwidth}
             \centering
             \includegraphics[width=1\linewidth]{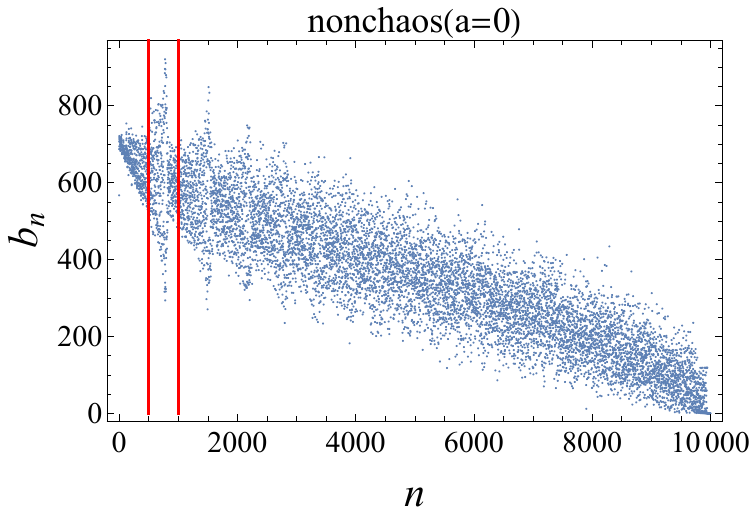} 
             \caption{$b_n$ in nonchaotic case $a=0$.}
             \label{fig:B-bn_GUE_nonchaos}
        \end{subfigure}     
        \caption{ The Lanczos coefficients  $\left \{ b_n \right \} $ for an initial operator selected from GUE. The sets of coefficients between the two red lines ($500\le n\le 1000$) are used to calculate the variance $\sigma^2$~\eqref{eq:variance}. }
        \label{fig:B-bn-example-100}
\end{figure}
% \begin{figure}[h][H]
% \centering
% \includegraphics[width=0.7\textwidth]{figures/bn-chaos-GUE.pdf} 
% \caption{$b_n$ behavior chaos GUE. The region between the two red lines is used to calculate the $\sigma^2$}
% \label{fig:B-bn_GUE_chaos}
% \end{figure}

% \begin{figure}[h][H]
% \centering
% \includegraphics[width=0.7\textwidth]{figures/bn-nonchaos-GUE.pdf} 
% \caption{$b_n$ behavior non-chaos GUE. The region between the two red lines is used to calculate the $\sigma^2$}
% \label{fig:B-bn_GUE_nonchaos}
% \end{figure}

\begin{figure}[h]
        \centering
        \begin{subfigure}[b]{0.48\textwidth}
                \centering
                \includegraphics[width=\linewidth]{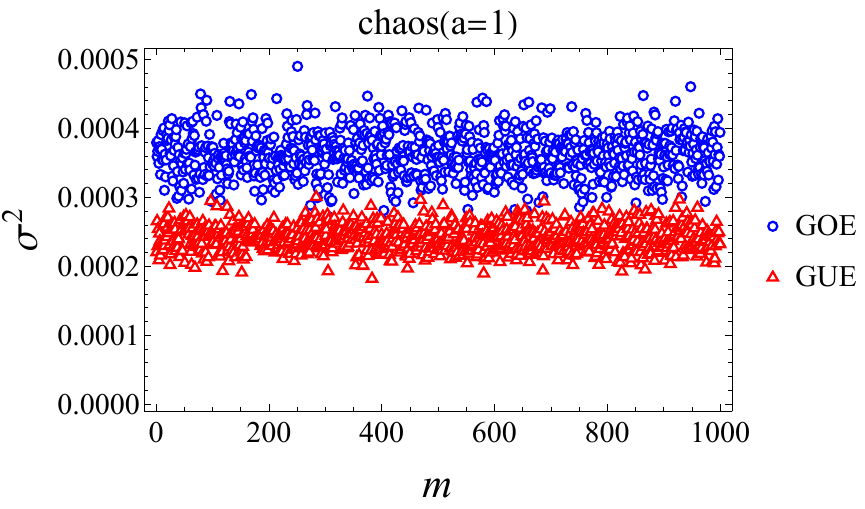}
                \caption{Sampled variances $\sigma^2$ of GOE and GUE in chaotic case $a=1$. }
                \label{fig:B-Gaussian-100-chaos}
        \end{subfigure}
           \hfill
        \begin{subfigure}[b]{0.48\textwidth}
                \centering
                \includegraphics[width=\linewidth]{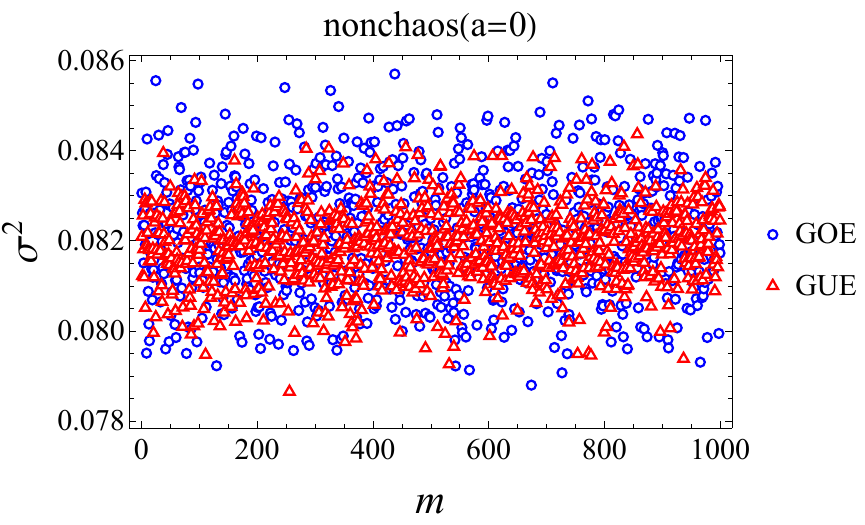}
                 \caption{Sampled variances $\sigma^2$ of GOE and GUE in nonchaotic case $a=0$. }
                 \label{fig:B-Gaussian-100-nonchaos}
        \end{subfigure}
        \begin{subfigure}[b]{0.48\textwidth}
                \centering
                \includegraphics[width=\linewidth]{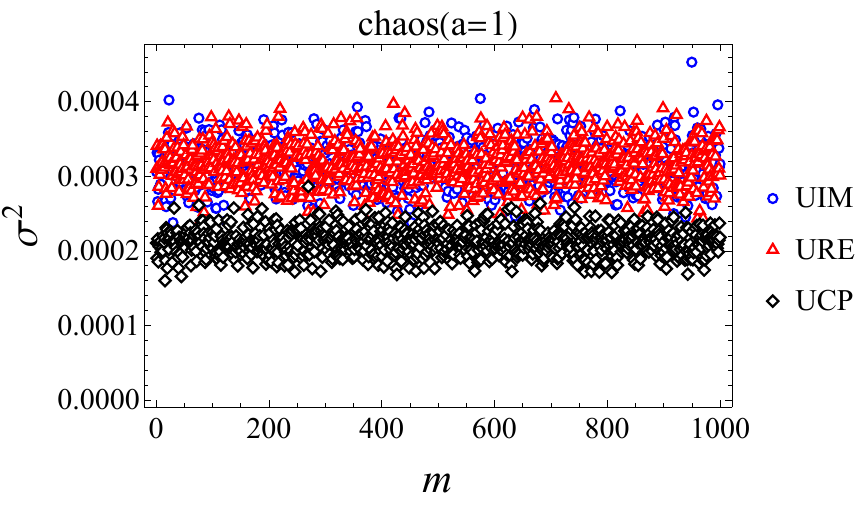}
                \caption{Sampled variances $\sigma^2$ of  uniform distributions in chaotic case $a=1$.}
                \label{fig:B-uni-sigma-100-chaos}
        \end{subfigure}
           \hfill
        \begin{subfigure}[b]{0.48\textwidth}
                \centering
                \includegraphics[width=\linewidth]{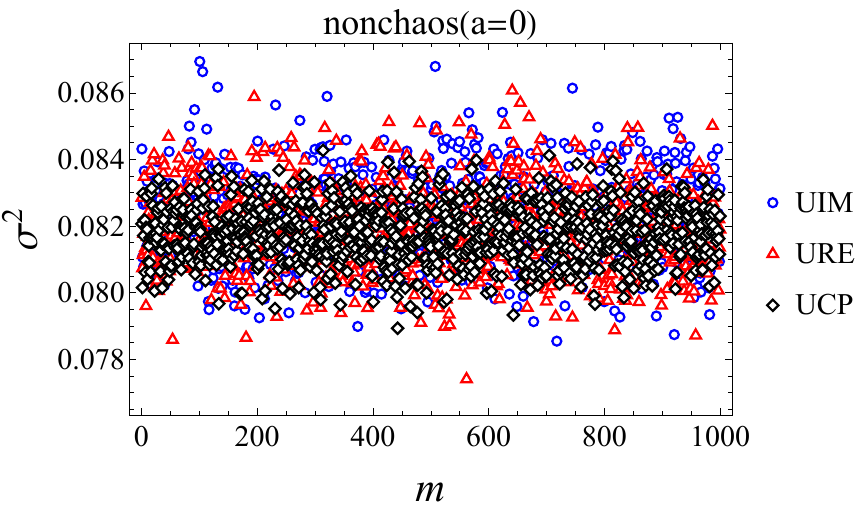}
                 \caption{Sampled variances $\sigma^2$ of uniform distributions in nonchaotic case $a=0$.}
                 \label{fig:B-uni-sigma-100-nonchaos}
        \end{subfigure}        
        \caption{Samples of  variances $\sigma^2$ for GOE, GUE, URE, UIM and UCP.  The horizontal axis represents the $m$th sampling. The chaotic $a=1$ and nonchatic case $a=0$ have different behaviors.  }
        \label{fig:billiard-sigmaG-behavior-100}
\end{figure}

\begin{figure}[h]
        \centering
        \begin{subfigure}[b]{0.48\textwidth}
                \centering
                \includegraphics[width=\linewidth]{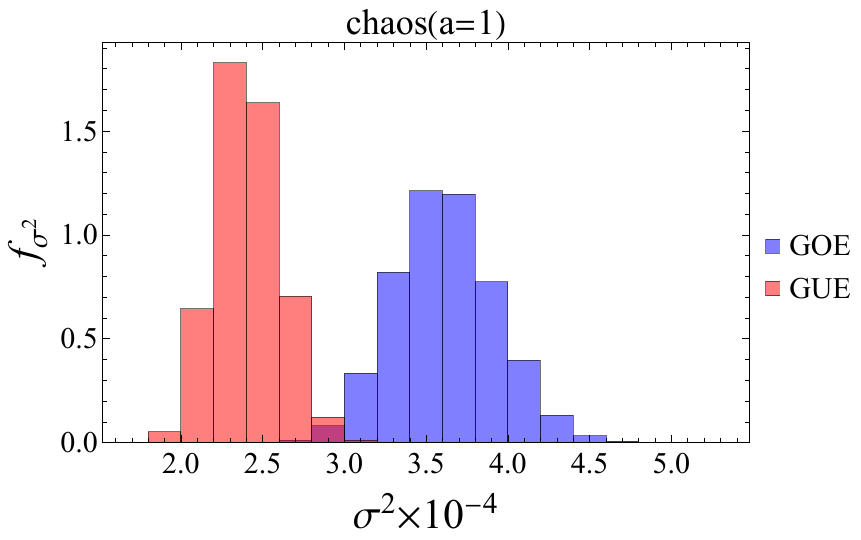}
                 \caption{Histogram of $\sigma^2$ of GOE and GUE in the chaotic case $a=1$.}
                 \label{fig:B-GaussianHist-100-chaos}
        \end{subfigure}
           \hfill
        \begin{subfigure}[b]{0.48\textwidth}
                \centering
                \includegraphics[width=\linewidth]{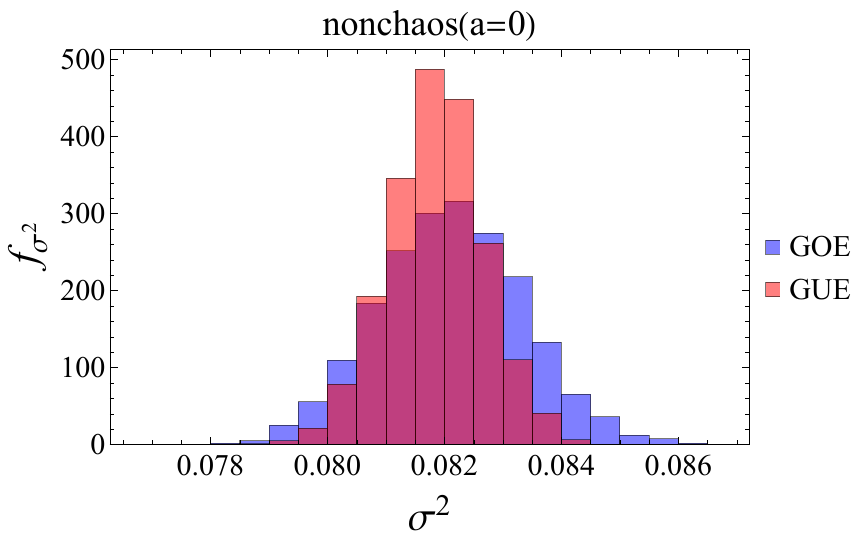}
                \caption{Histogram of $\sigma^2$ of GOE and GUE in the nonchaotic case $a=0$.}
                \label{fig:B-GaussianHist-100-nonchaos}
        \end{subfigure}
        \begin{subfigure}[b]{0.48\textwidth}
                \centering
                \includegraphics[width=\linewidth]{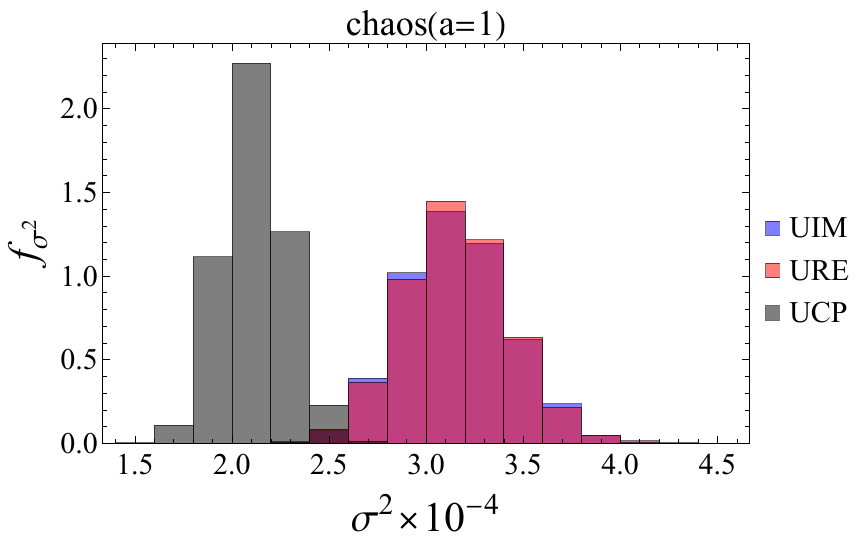}
                 \caption{Histogram of $\sigma^2$ of  uniform distributions in the chaotic case $a=1$.}
                 \label{fig:B-uniformHist-100-chaos}
        \end{subfigure}
           \hfill
        \begin{subfigure}[b]{0.48\textwidth}
                \centering
                \includegraphics[width=\linewidth]{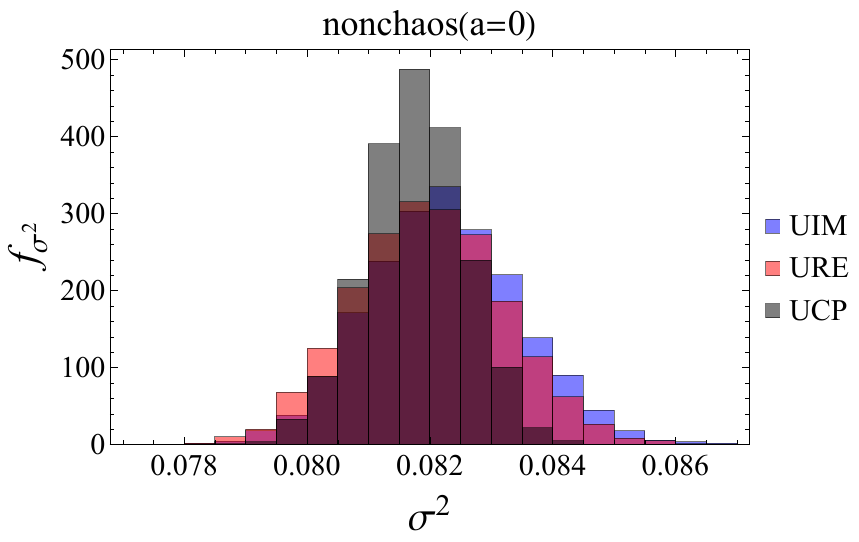}
                \caption{Histogram of $\sigma^2$ of  uniform distributions in the nonchaotic case $a=0$.}
                \label{fig:B-uniformHist-100-nonchaos}
        \end{subfigure}        
        \caption{The histogram of variances $\sigma^2$ after sampling the initial operator 5000 times from GOE, GUE, URE, UIM and UCP. The property of being normal and the overlapping-separating behavior are obvious. }
        \label{fig:billiard-Histdistribution-100}
\end{figure}

\begin{figure}[h]
        \centering
        \begin{subfigure}[b]{0.48\textwidth}
                \centering
                \includegraphics[width=\linewidth]{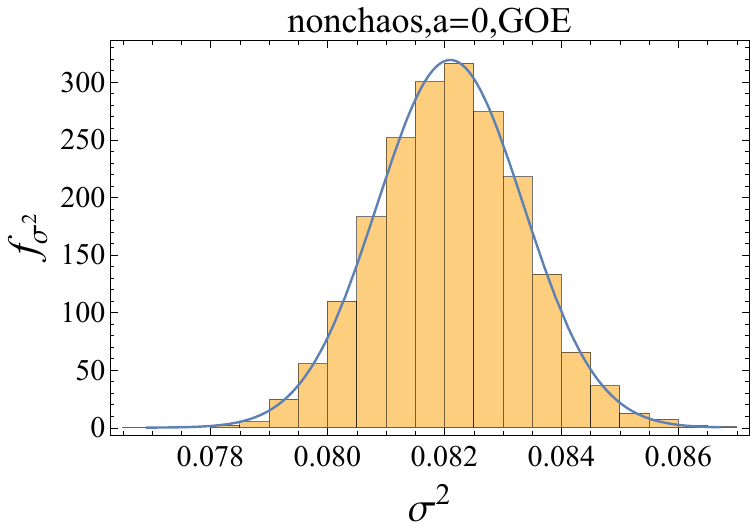}
                 \caption{Nonchaotic case $a=0$  of GOE. Skewness=0.0592927, Kurtosis=3.02995}
                 \label{fig:B-GOE-fit-nonchaos-100}
        \end{subfigure}
           \hfill
        \begin{subfigure}[b]{0.48\textwidth}
                \centering
                \includegraphics[width=\linewidth]{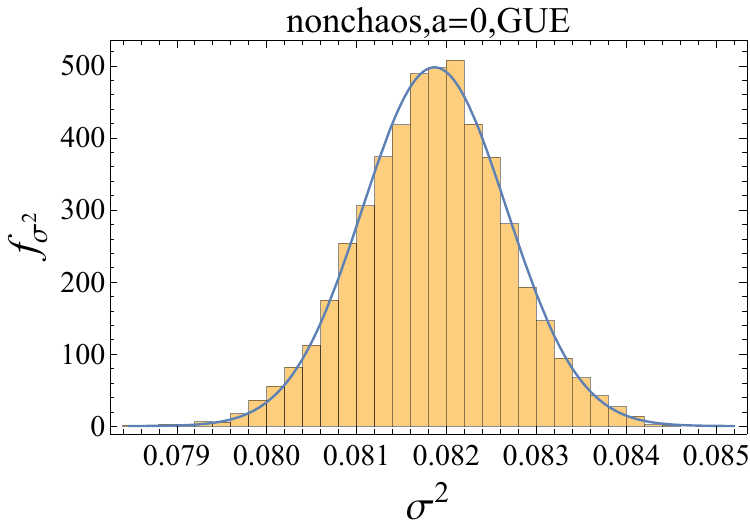}
                 \caption{Nonchaotic case $a=0$  of GUE. Skewness=-0.06833, Kurtosis=3.13681}
                 \label{fig:B-GUE-fit-nonchaos-100}
        \end{subfigure}
        \begin{subfigure}[b]{0.48\textwidth}
                \centering
                \includegraphics[width=\linewidth]{figures/100NchaosGOEfit.pdf}
                 \caption{Chaotic case $a=1$  of GOE. Skewness=0.201339, Kurtosis=3.19744}
                 \label{fig:B-GOE-fit-chaos-100}
        \end{subfigure}
           \hfill
        \begin{subfigure}[b]{0.48\textwidth}
                \centering
                \includegraphics[width=\linewidth]{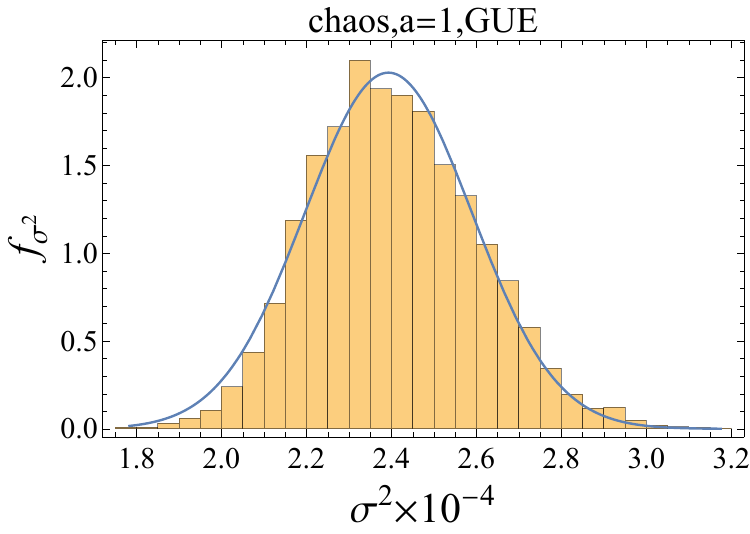}
                 \caption{Chaotic case $a=1$  of GUE. Skewness=0.247936, Kurtosis=2.98557}
                 \label{fig:B-GUE-fit-chaos-100}
        \end{subfigure}
        \caption{Fitted normal distribution for GOE and GUE, $\mathcal{N}_{max}=100$.}
        \label{fig:billiard-fitted-distribution-Gaussian-100}
\end{figure}

\begin{figure}[h]
        \centering
        \begin{subfigure}[b]{0.48\textwidth}
                \centering
                \includegraphics[width=\linewidth]{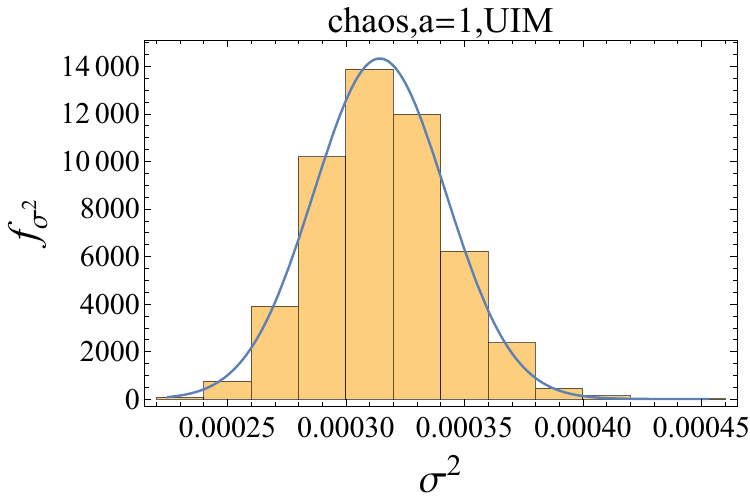}
                 \caption{Chaotic case $a=1$  of UIM. Skewness=0.225317, Kurtosis=3.18781.}
                 \label{fig:B-uni-fit-chaosIM-100}
        \end{subfigure}
           \hfill
        \begin{subfigure}[b]{0.48\textwidth}
                \centering
                \includegraphics[width=\linewidth]{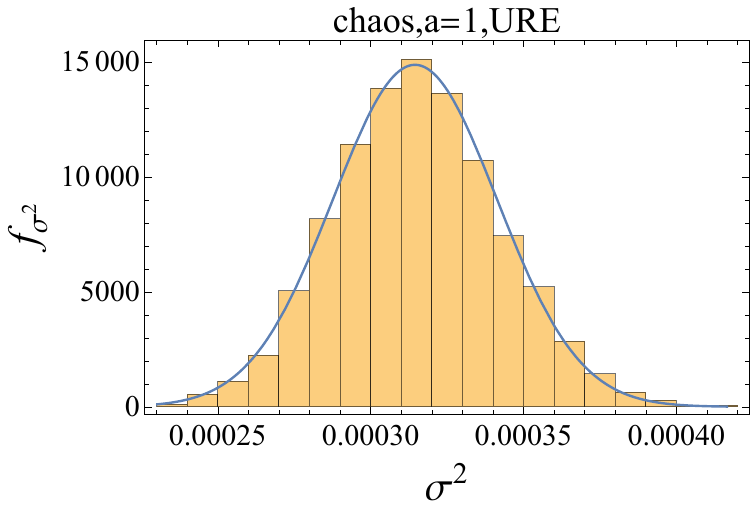}
                 \caption{Chaotic case $a=1$  of URE. Skewness=0.131034, Kurtosis=3.00302.}
                 \label{fig:B-uni-fit-chaosRE-100}
        \end{subfigure}
        \begin{subfigure}[b]{0.48\textwidth}
                \centering
                \includegraphics[width=\linewidth]{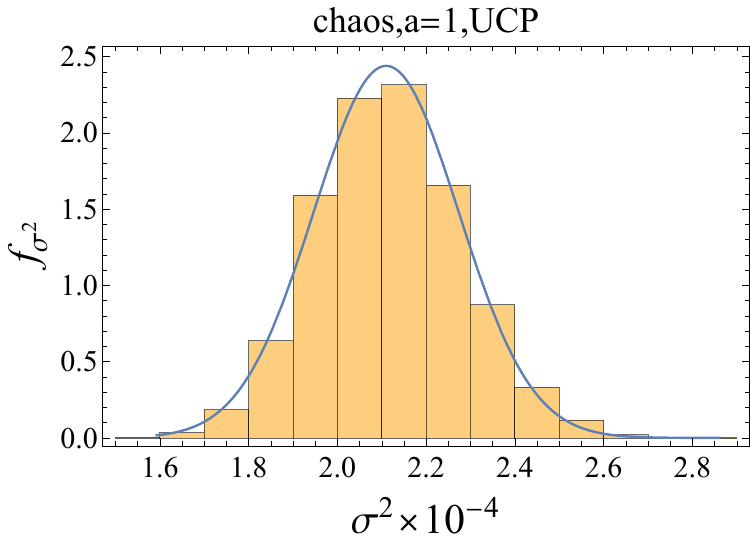}
                 \caption{Chaotic case $a=1$  of UCP. Skewness=0.167094, Kurtosis=3.06099.}
                 \label{fig:uni-fit-chaosCP-100}
        \end{subfigure}
           \hfill
        \begin{subfigure}[b]{0.48\textwidth}
                \centering
                \includegraphics[width=\linewidth]{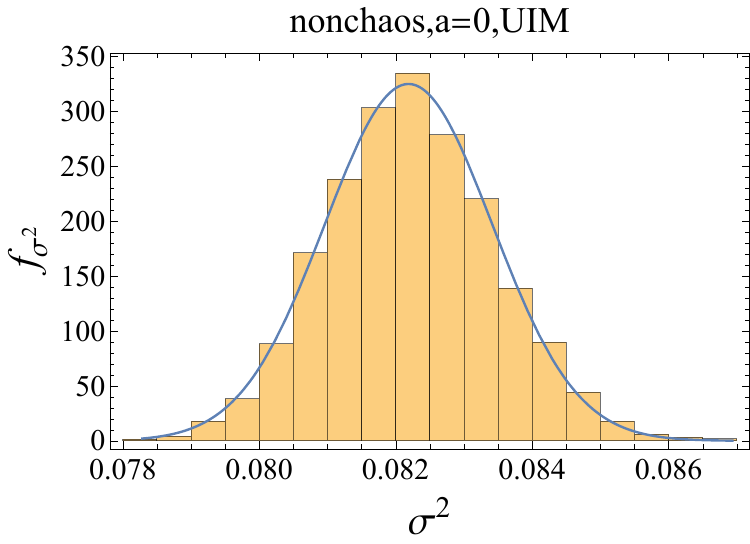}
                 \caption{Nonchaotic case $a=0$  of UIM. Skewness=0.1411, Kurtosis=3.08718.}
                 \label{fig:uni-fit-nonchaosIM-100}
        \end{subfigure}
        \begin{subfigure}[b]{0.48\textwidth}
                \centering
                \includegraphics[width=\linewidth]{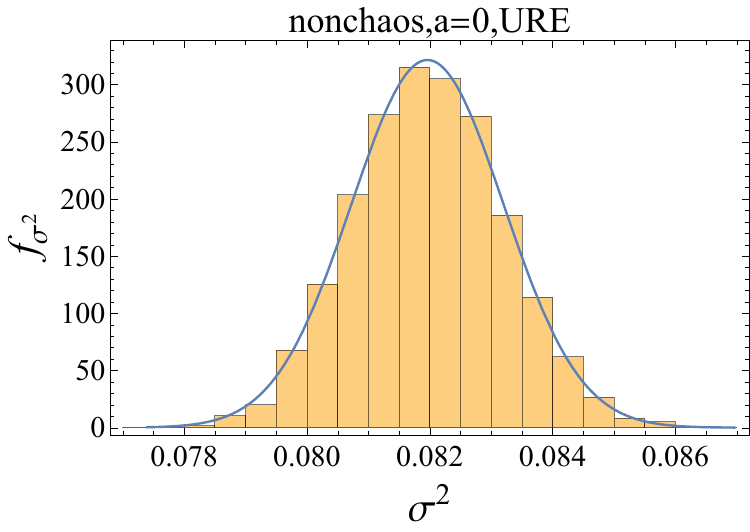}
                 \caption{Nonchaotic case $a=0$  of URE. Skewness=0.0823458, Kurtosis=2.97123.}
                 \label{fig:uni-fit-nonchaosRE-100}
        \end{subfigure}
           \hfill
        \begin{subfigure}[b]{0.48\textwidth}
                \centering
                \includegraphics[width=\linewidth]{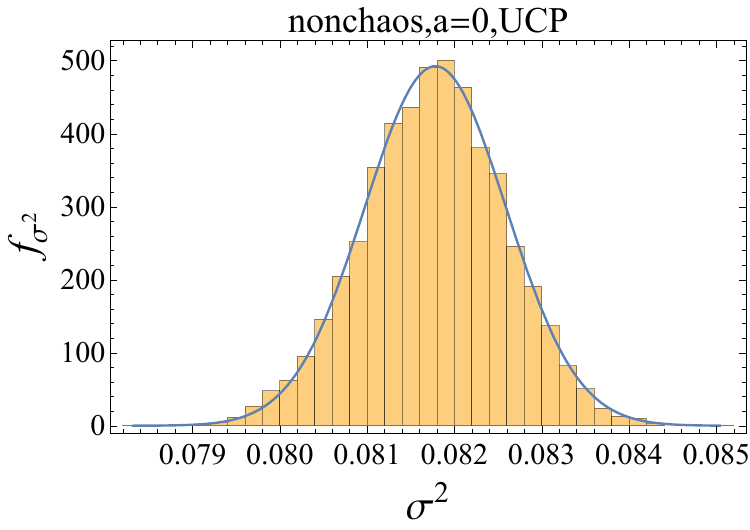}
                 \caption{Nonchaotic case $a=0$  of UCP. Skewness=-0.0801398, Kurtosis=3.04548.}
                 \label{fig:uni-fit-nonchaosCP-100}
        \end{subfigure}
        \caption{Fitted normal distribution for uniform distributions, $\mathcal{N}_{max}=100$.}
        \label{fig:billiard-fitted-distribution-uniform-100}
\end{figure}

\begin{table}[!htbp]
\centering
\begin{tabular}{ccccc}
\toprule
\multicolumn{3}{c}{\multirow{2}{*}{}} & \multicolumn{2}{c}{\textbf{$\mathcal{N}_{max}=100$}}\\
\cmidrule{4-5}
\multicolumn{3}{c}{}  & $\mu_0$ & $\sigma_0$ \\
\midrule 
\multirow{10}{*}{\textbf{nonchaos}} & \multirow{2}{*}{GOE}  & Fit & 0.08205   & 0.00125033 \\
 &    & Data   &   0.0821079   & 0.00124627\\
 \cmidrule{2-5}
 & \multirow{2}{*}{GUE}    & Fit   & 0.0818423   & 0.000799776\\
 &    & Data  &   0.0818522   & 0.000818971\\
 \cmidrule{2-5}
 & \multirow{2}{*}{URE}    & Fit   & 0.0819622 & 0.0012411\\
 &    & Data  &   0.0819846 & 0.0012243\\
  \cmidrule{2-5}
 & \multirow{2}{*}{UIM}    & Fit   & 0.0821883 & 0.00122778\\
 &    & Data   &   0.0822335 & 0.00124233\\
  \cmidrule{2-5}
 & \multirow{2}{*}{UCP}    & Fit   & 0.0817821& 0.000810266\\
 &    & Data   &   0.0817644& 0.000816963\\
\midrule 
\multirow{10}{*}{\textbf{chaos}} & \multirow{2}{*}{GOE}  & Fit  &  0.000360084 & 0.0000311866 \\
 &    & Data  &  0.000361536   & 0.0000312431\\
 \cmidrule{2-5}
 & \multirow{2}{*}{GUE}    & Fit  &   0.000239301 & 0.0000196699\\
 &    & Data  &   0.00024081   & 0.0000194523\\
 \cmidrule{2-5}
 & \multirow{2}{*}{URE}    & Fit   &   0.000314526& 0.0000268396\\
 &    & Data   &   0.000315565& 0.0000270371\\
 \cmidrule{2-5}
 & \multirow{2}{*}{UIM}    & Fit  &   0.000314332& 0.0000278555\\
 &    & Data   &   0.000315618& 0.0000277744\\
 \cmidrule{2-5}
 & \multirow{2}{*}{UCP}    & Fit   &   2.11016 & 0.163631\\
 &    & Data   &   2.11775 & 0.164331\\
\bottomrule
\end{tabular}
\caption{Average $\mu_0$ and standard deviation $\sigma_0$ for the data fitting, $\mathcal{N}_{max}=100$. 'Data' means the values computed from the sampled data. 'Fit' means the values of the fitted normal distribution. }
\label{tab:billiard-fitted-data-100}
\end{table}

\begin{figure}[h]
        \centering
        \begin{subfigure}[b]{0.48\textwidth}
                \centering
                \includegraphics[width=\linewidth]{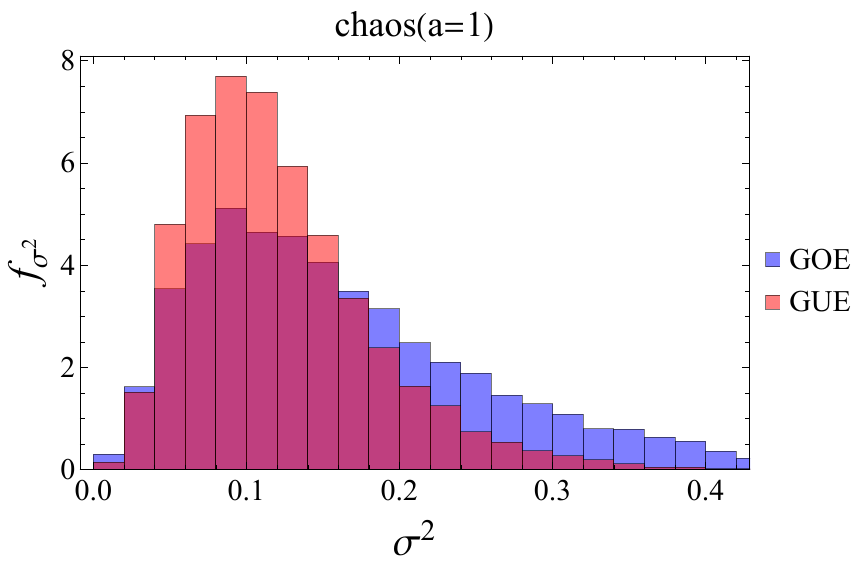}
                 \caption{Histogram of $\sigma^2$ of GOE and GUE in the chaotic case $a=1$.}
                 \label{fig:B-GaussianHist-5-chaos}
        \end{subfigure}
           \hfill
        \begin{subfigure}[b]{0.48\textwidth}
                \centering
                \includegraphics[width=\linewidth]{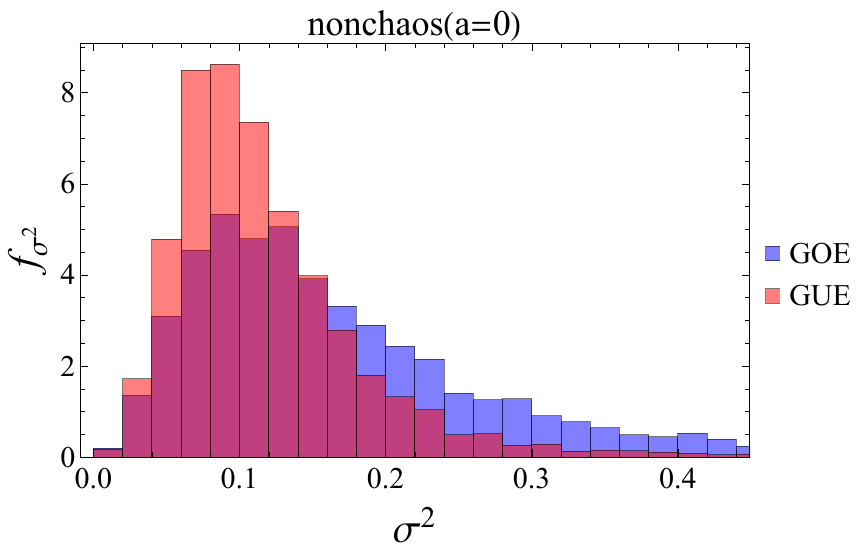}
                \caption{Histogram of $\sigma^2$ of GOE and GUE in the nonchaotic case $a=0$.}
                \label{fig:B-GaussianHist-5-nonchaos}
        \end{subfigure}
        \begin{subfigure}[b]{0.48\textwidth}
                \centering
                \includegraphics[width=\linewidth]{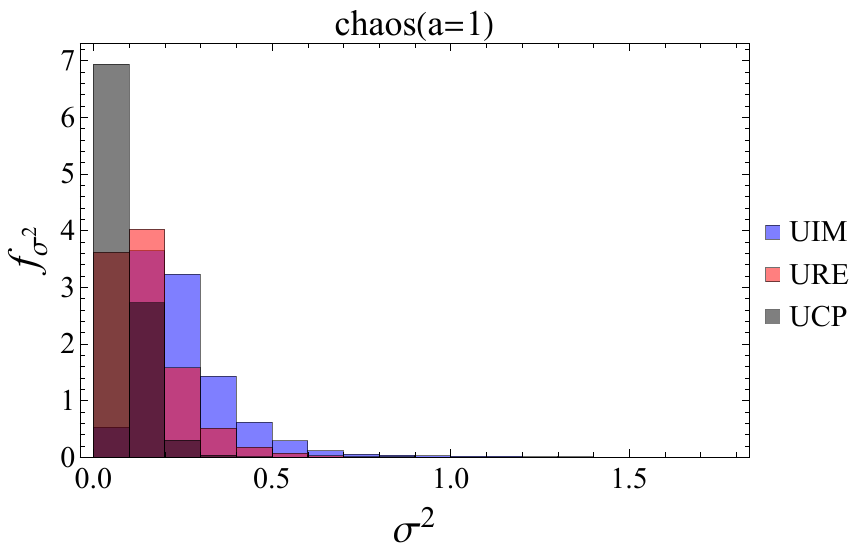}
                 \caption{Histogram of $\sigma^2$ of  uniform distributions in the chaotic case $a=1$.}
                 \label{fig:B-uniformHist-5-chaos}
        \end{subfigure}
           \hfill
        \begin{subfigure}[b]{0.48\textwidth}
                \centering
                \includegraphics[width=\linewidth]{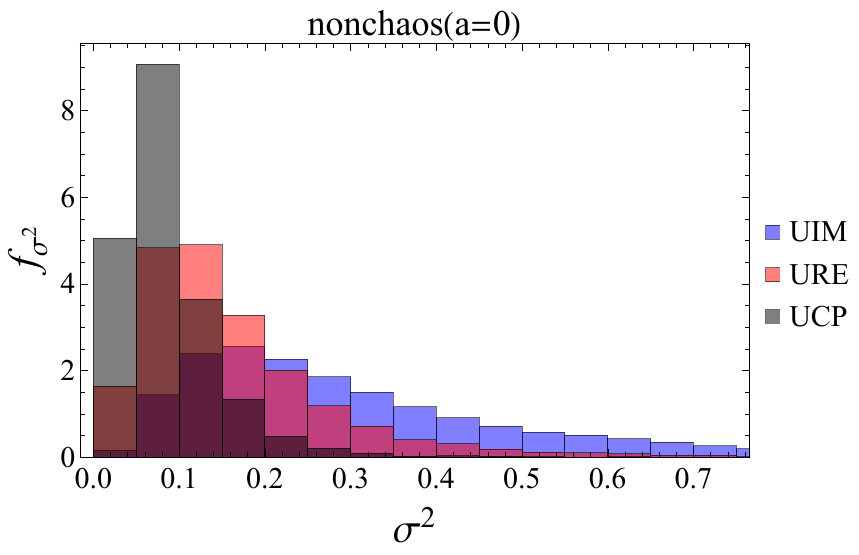}
                \caption{Histogram of $\sigma^2$ of  uniform distributions in the nonchaotic case $a=0$.}
                \label{fig:B-uniformHist-5-nonchaos}
        \end{subfigure}
        \caption{The histogram of variances $\sigma^2$ after sampling the initial operator 5000 times from GOE, GUE, URE, UIM and UCP, with $\mathcal{N}_{max}=5$. Obviously they are not normal distributions.}
        \label{fig:billiard-hist-distribution-GE-5}
\end{figure}

\begin{figure}[h]
        \centering
        \begin{subfigure}[b]{0.48\textwidth}
                \centering
                \includegraphics[width=\linewidth]{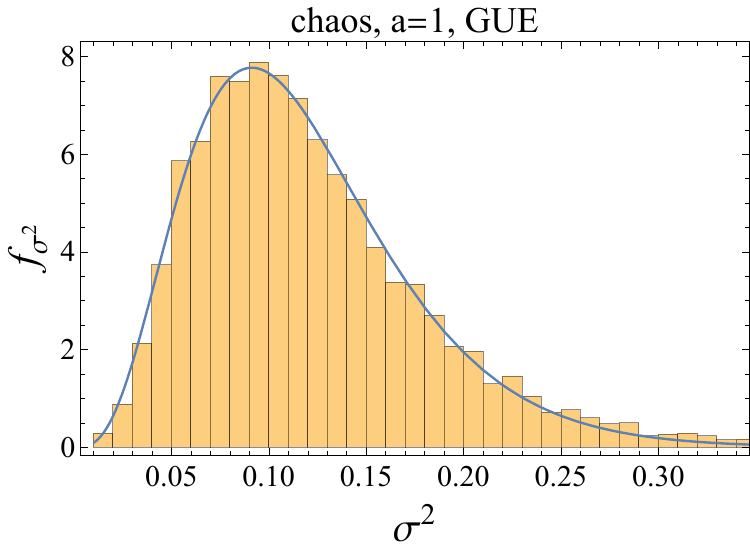}
                 \caption{$N=5$ chaos, GUE. Skewness=1.19843, Kurtosis=5.18562.}
                 \label{fig:B-GaussianHistGUE-fit-5-chaos}
        \end{subfigure}
           \hfill
        \begin{subfigure}[b]{0.48\textwidth}
                \centering
                \includegraphics[width=\linewidth]{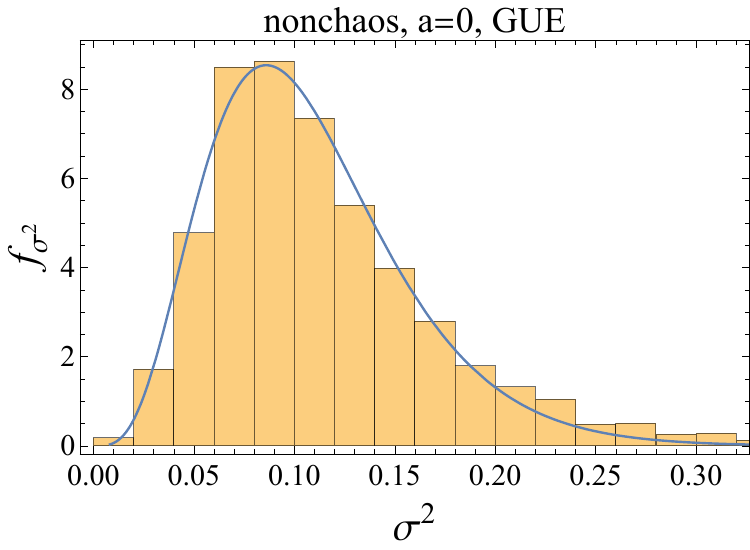}
                \caption{$N=5$ nonchaos, GUE. Skewness=4.30801, Kurtosis=50.7359.}
                \label{fig:B-GaussianHistGUE-fit-5-nonchaos}
        \end{subfigure}
        \begin{subfigure}[b]{0.48\textwidth}
                \centering
                \includegraphics[width=\linewidth]{figures/5N_GOE_fit_chaos.pdf}
                 \caption{$N=5$ chaos, GOE. Skewness=1.6231, Kurtosis=7.05429.}
                 \label{fig:B-GaussianHistGOE-fit-5-chaos}
        \end{subfigure}
           \hfill
        \begin{subfigure}[b]{0.48\textwidth}
                \centering
                \includegraphics[width=\linewidth]{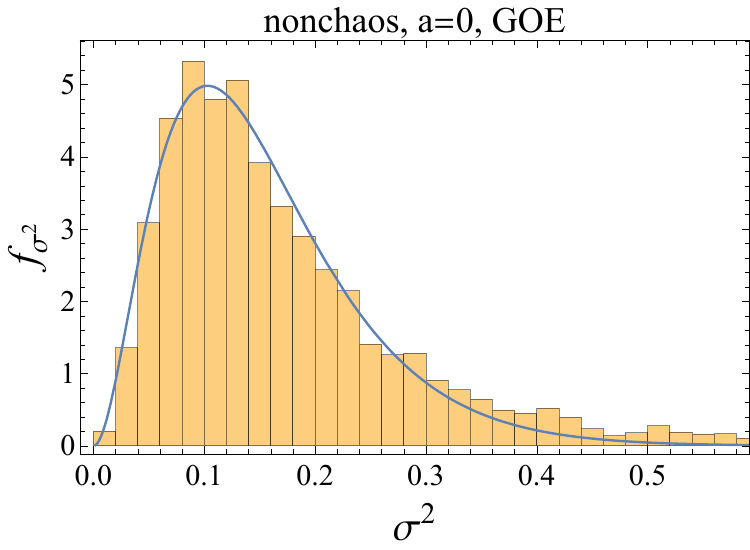}
                \caption{$N=5$ nonchaos, GOE. Skewness=5.83862, Kurtosis=73.6201.}
                \label{fig:B-GaussianHistGOE-fit-5-nonchaos}
        \end{subfigure}
        \caption{Data fitting of rescaled chi-square distribution for GOE and GUE, $\mathcal{N}_{max}=5$. }
        \label{fig:billiard-his-distribution-fit-GE-5}
\end{figure}

\begin{figure}[h]
        \centering
        \begin{subfigure}[b]{0.48\textwidth}
                \centering
                \includegraphics[width=\linewidth]{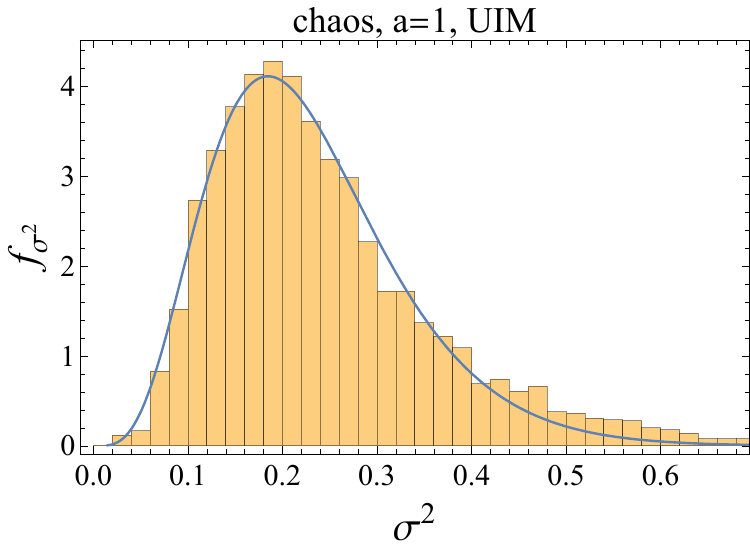}
                 \caption{Chaotic case $a=1$  of UIM. Skewness=0.306957, Kurtosis=3.26486.}
                 \label{fig:B-uni-fit-chaosIM-5}
        \end{subfigure}
           \hfill
        \begin{subfigure}[b]{0.48\textwidth}
                \centering
                \includegraphics[width=\linewidth]{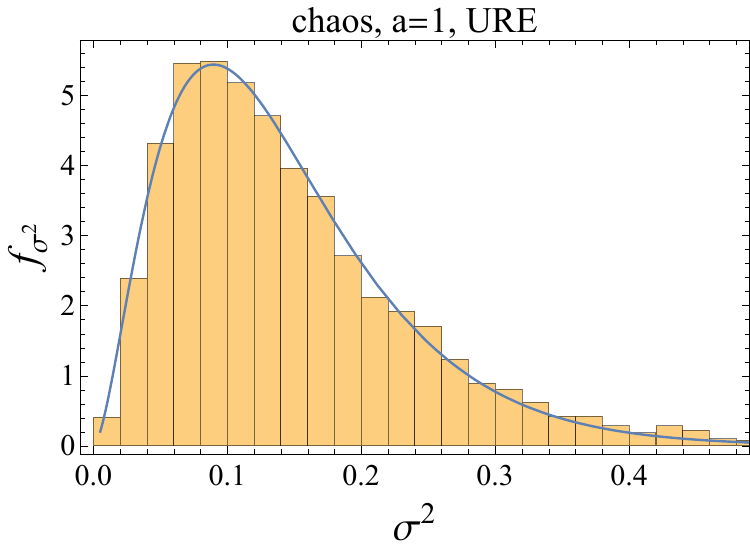}
                 \caption{Chaotic case $a=1$  of URE. Skewness=0.234259, Kurtosis=2.88038.}
                 \label{fig:B-uni-fit-chaosRE-5}
        \end{subfigure}
        \begin{subfigure}[b]{0.48\textwidth}
                \centering
                \includegraphics[width=\linewidth]{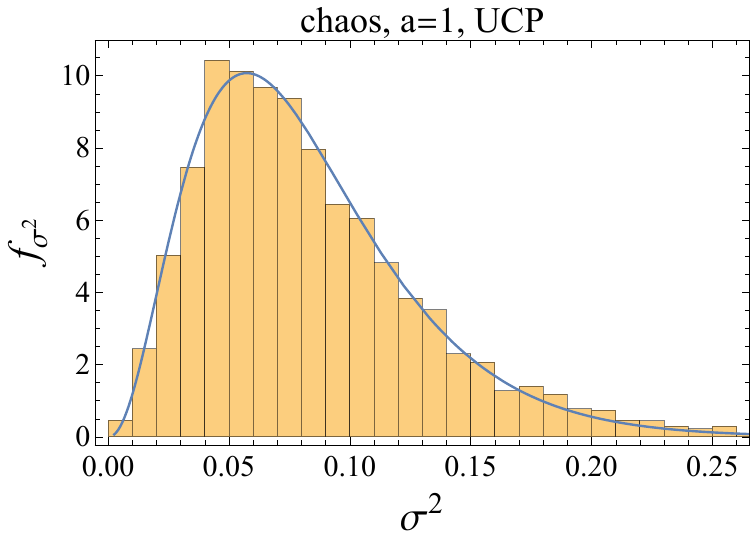}
                 \caption{Chaotic case $a=1$  of UCP. Skewness=0.207669, Kurtosis=2.97609.}
                 \label{fig:uni-fit-chaosCP-5}
        \end{subfigure}x
           \hfill
        \begin{subfigure}[b]{0.48\textwidth}
                \centering
                \includegraphics[width=\linewidth]{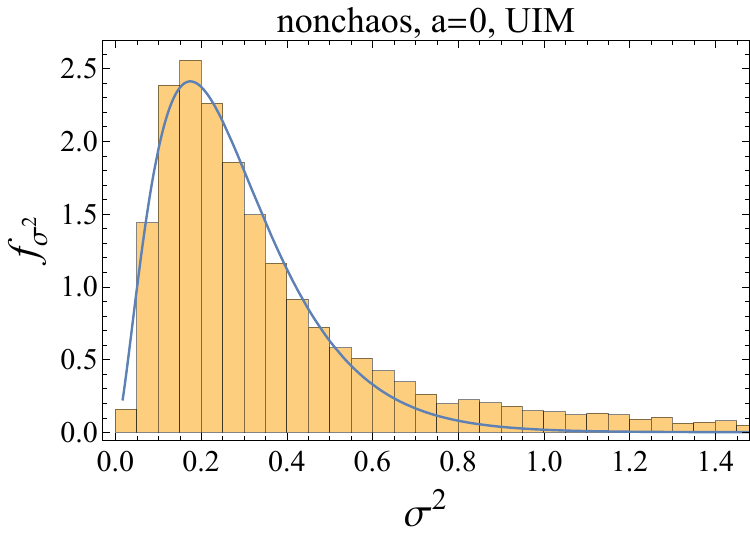}
                 \caption{Nonchaotic case $a=0$  of UIM. Skewness=0.0715752, Kurtosis=2.99478.}
                 \label{fig:uni-fit-nonchaosIM-5}
        \end{subfigure}
        \begin{subfigure}[b]{0.48\textwidth}
                \centering
                \includegraphics[width=\linewidth]{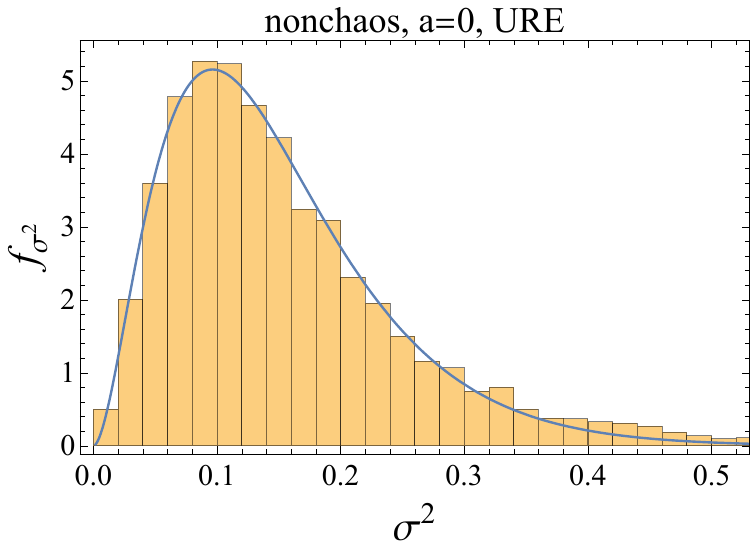}
                 \caption{Nonchaotic case $a=0$  of URE. Skewness=0.0339498, Kurtosis=3.07535.}
                 \label{fig:uni-fit-nonchaosRE-5}
        \end{subfigure}
           \hfill
        \begin{subfigure}[b]{0.48\textwidth}
                \centering
                \includegraphics[width=\linewidth]{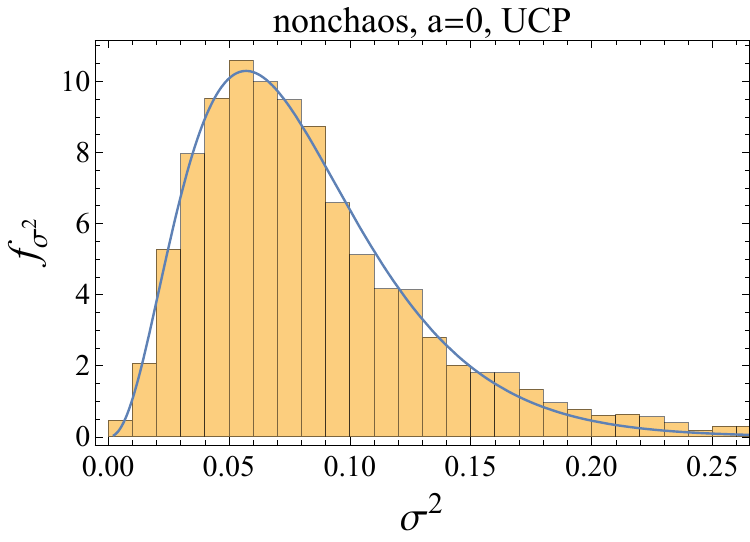}
                 \caption{Nonchaotic case $a=0$  of UCP. Skewness=-0.135062, Kurtosis=3.25632.}
                 \label{fig:uni-fit-nonchaosCP-5}
        \end{subfigure}
        \caption{Data fitting of rescaled chi-square distribution for uniform distributions, $\mathcal{N}_{max}=5$.}
        \label{fig:billiard-fitted-distribution-uniform-5}
\end{figure}

\begin{figure}[h]
        \centering
        \begin{subfigure}[b]{0.48\textwidth}
                \centering
                \includegraphics[width=\linewidth]{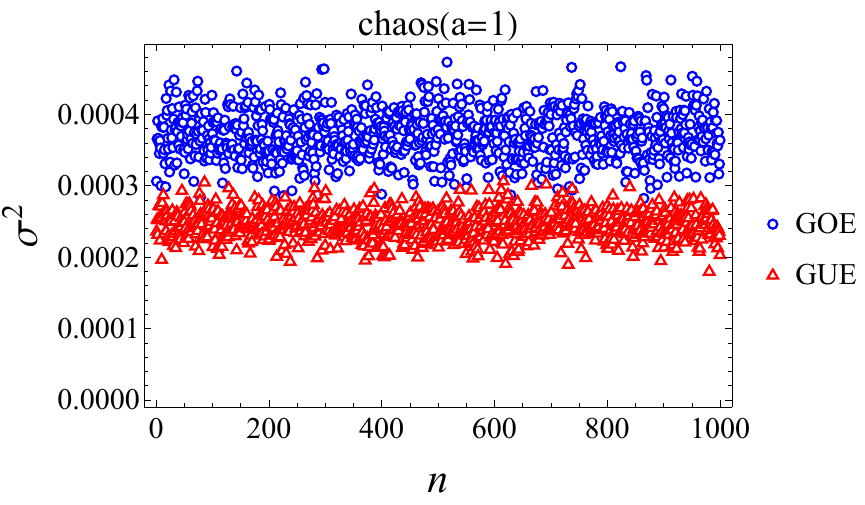}
                \caption{Sampled variances $\sigma^2$ of GOE and GUE in chaotic case $a=1$. }
                \label{fig:B-sigmaG-choice-100-chaos}
        \end{subfigure}
           \hfill
        \begin{subfigure}[b]{0.48\textwidth}
                \centering
                \includegraphics[width=\linewidth]{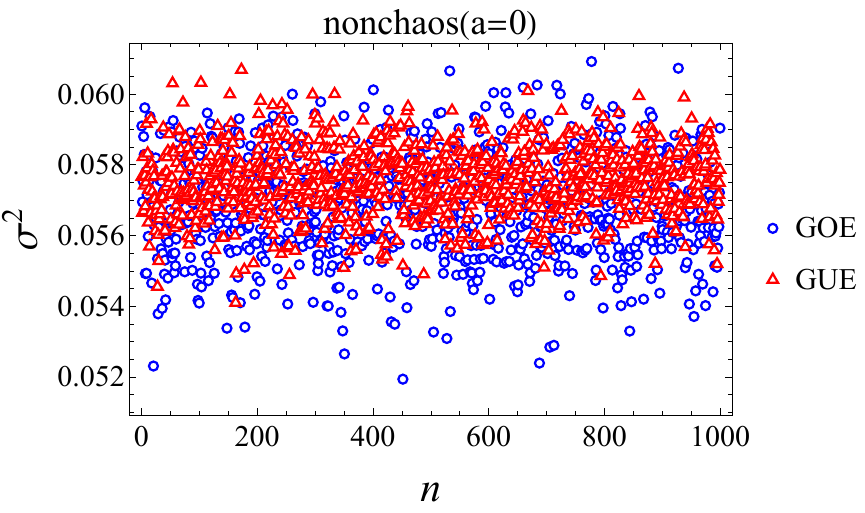}
                 \caption{Sampled variances $\sigma^2$ of GOE and GUE in nonchaotic case $a=0$. }
                 \label{fig:B-sigmaG-choice-100-nonchaos}
        \end{subfigure}
        \begin{subfigure}[b]{0.48\textwidth}
                \centering
                \includegraphics[width=\linewidth]{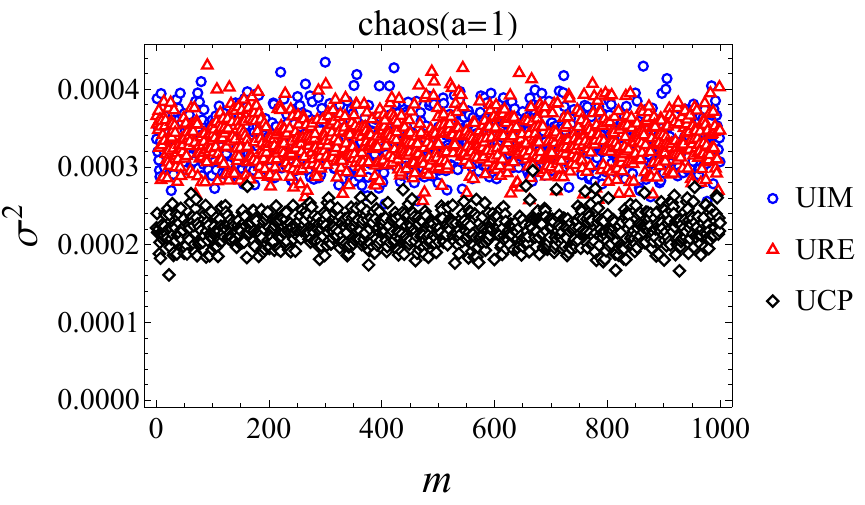}
                \caption{Sampled variances $\sigma^2$ of  uniform distributions in chaotic case $a=1$.}
                \label{fig:B-uni-choice-sigma-100-chaos}
        \end{subfigure}
           \hfill
        \begin{subfigure}[b]{0.48\textwidth}
                \centering
                \includegraphics[width=\linewidth]{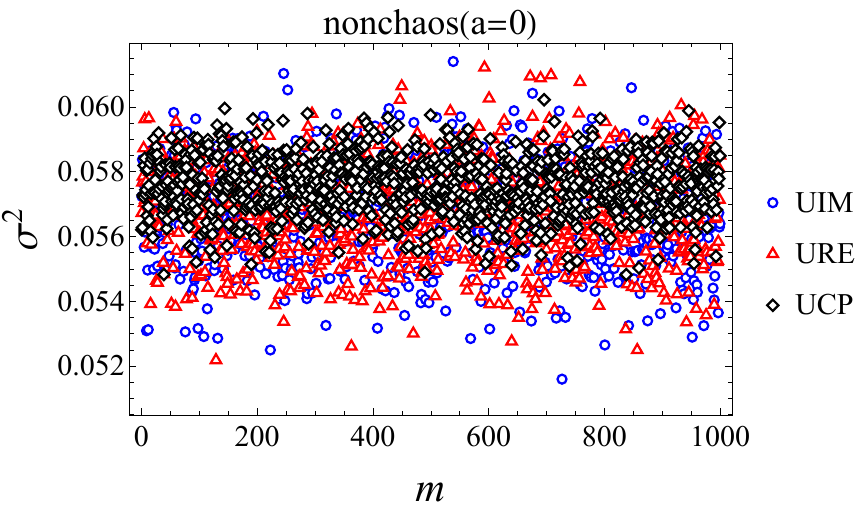}
                 \caption{Sampled variances $\sigma^2$ of uniform distributions in nonchaotic case $a=0$.}
                 \label{fig:B-uni-choice-sigma-100-nonchaos}
        \end{subfigure}
        \caption{The first $1000$ samples of  variances $\sigma^2$ for GOE, GUE, URE, UIM and UCP. Here the Lanczos coefficients  $\left \{ b_n \right \} $ from step-sizes $1000\le n\le 1500$ are used. }
        \label{fig:billiard-sigmaG-choice-100}
\end{figure}

\begin{figure}[h]
        \centering
        \begin{subfigure}[b]{0.48\textwidth}
                \centering
                \includegraphics[width=\linewidth]{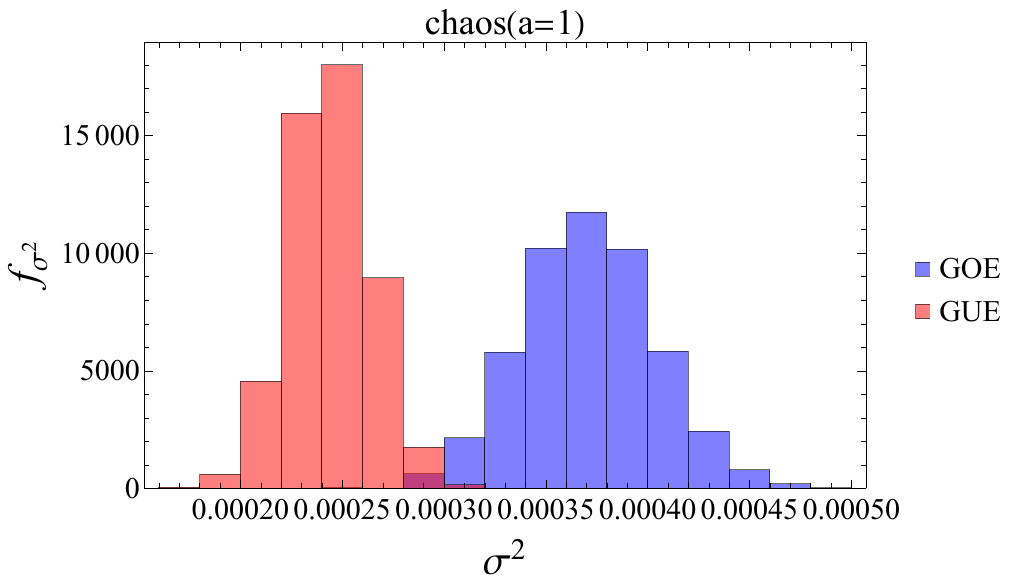}
                 \caption{Histogram of $\sigma^2$ of GOE and GUE in the chaotic case $a=1$.}
                 \label{fig:B-GaussianHist-choice-100-chaos}
        \end{subfigure}
           \hfill
        \begin{subfigure}[b]{0.48\textwidth}
                \centering
                \includegraphics[width=\linewidth]{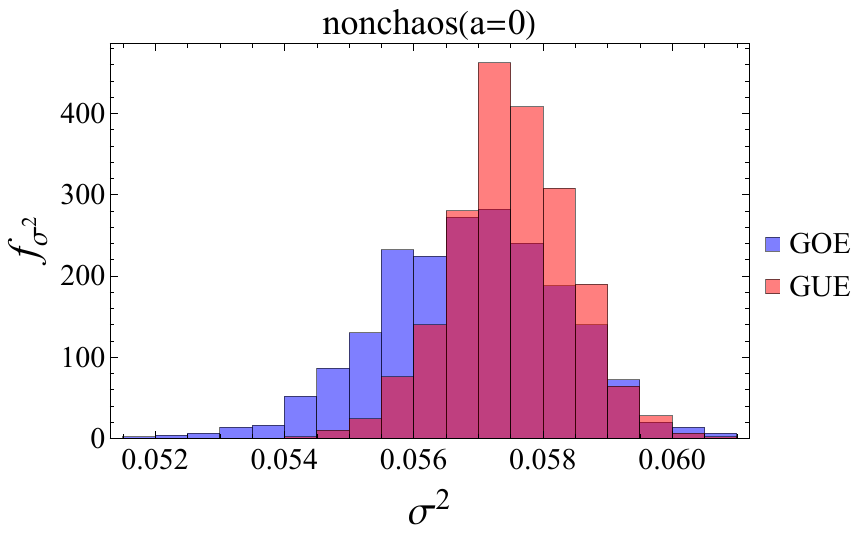}
                \caption{Histogram of $\sigma^2$ of GOE and GUE in the chaotic case $a=1$.}
                \label{fig:B-GaussianHist-choice-100-nonchaos}
        \end{subfigure}
        \begin{subfigure}[b]{0.48\textwidth}
                \centering
                \includegraphics[width=\linewidth]{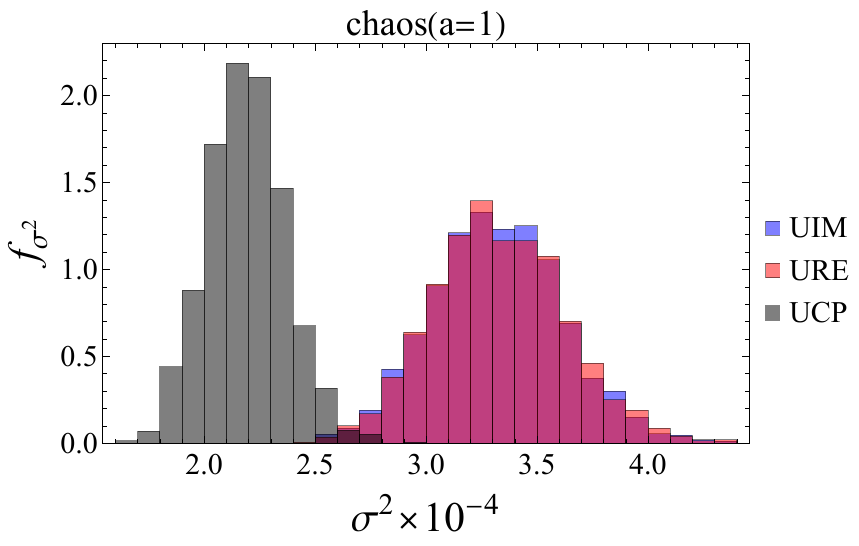}
                 \caption{Histogram of $\sigma^2$ of  uniform distributions in the chaotic case $a=1$.}
                 \label{fig:B-choice-uniformHist-100-chaos}
        \end{subfigure}
           \hfill
        \begin{subfigure}[b]{0.48\textwidth}
                \centering
                \includegraphics[width=\linewidth]{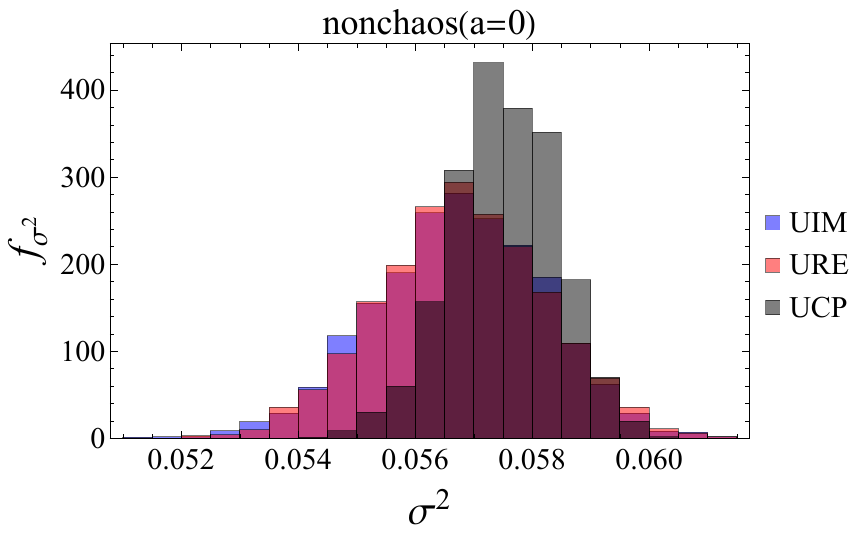}
                \caption{Histogram of $\sigma^2$ of  uniform distributions in the nonchaotic case $a=0$.}
                \label{fig:B-choice-uniformHist-100-nonchaos}
        \end{subfigure}        
        \caption{The histogram of variances $\sigma^2$ after sampling the initial operator 5000 times from GOE, GUE, URE, UIM and UCP. The Lanczos coefficients  $\left \{ b_n \right \} $ from step-sizes $1000\le n\le 1500$ are used. They are overlapped in the nonchaotic case and well-separated in the chaotic case.}
        \label{fig:billiard-Histdistribution-choice-100}
\end{figure}

\begin{figure}[h]
        \centering
        \begin{subfigure}[b]{0.48\textwidth}
                \centering
                \includegraphics[width=\linewidth]{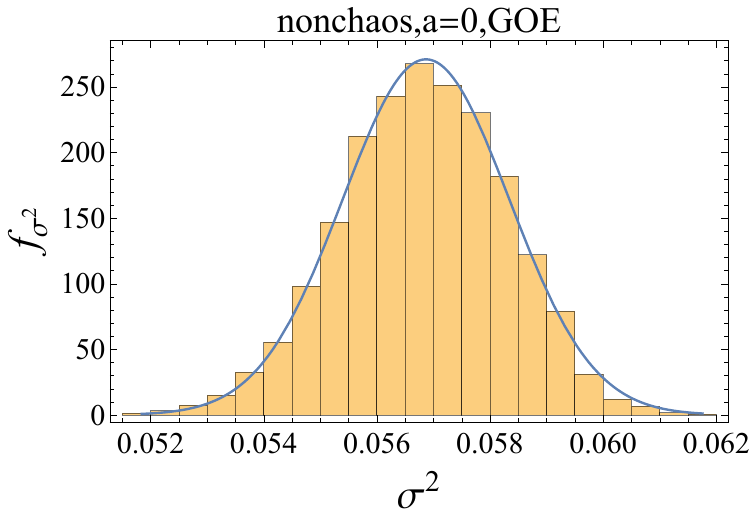}
                 \caption{Nonchaotic case $a=0$  of GOE. Skewness=-0.113535, Kurtosis=2.94238. }
                 \label{fig:B-GOE-fit-nonchaos-choice-100}
        \end{subfigure}
           \hfill
        \begin{subfigure}[b]{0.48\textwidth}
                \centering
                \includegraphics[width=\linewidth]{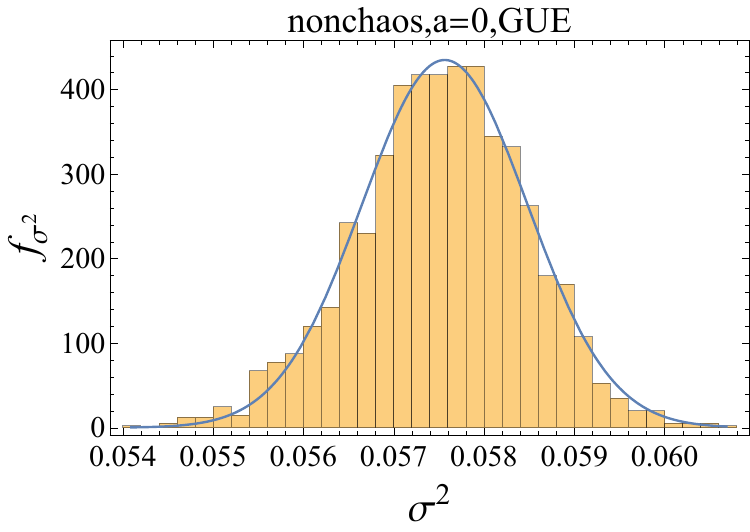}
                 \caption{Nonchaotic case $a=0$  of GUE. Skewness=-0.156533, Kurtosis=3.16475.}
                 \label{fig:B-GUE-fit-nonchaos-choice-100}
        \end{subfigure}
        \begin{subfigure}[b]{0.48\textwidth}
                \centering
                \includegraphics[width=\linewidth]{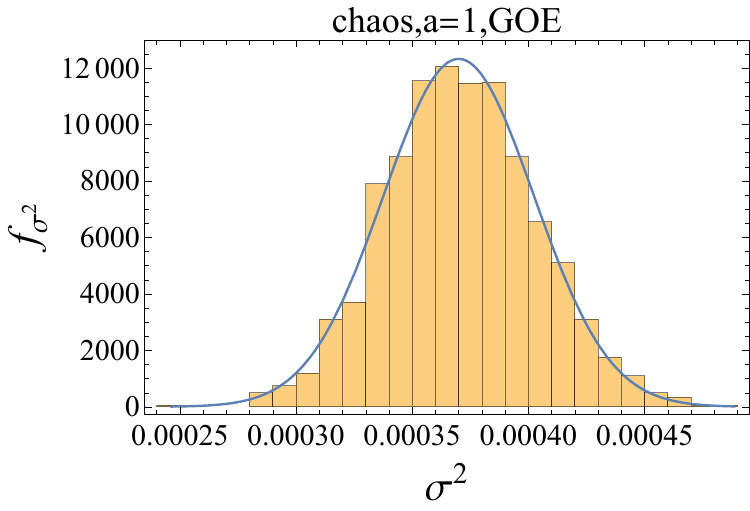}
                 \caption{Chaotic case $a=1$  of GOE. Skewness=0.107514, Kurtosis=3.01354.}
                 \label{fig:B-GOE-fit-chaos-choice-100}
        \end{subfigure}
           \hfill
        \begin{subfigure}[b]{0.48\textwidth}
                \centering
                \includegraphics[width=\linewidth]{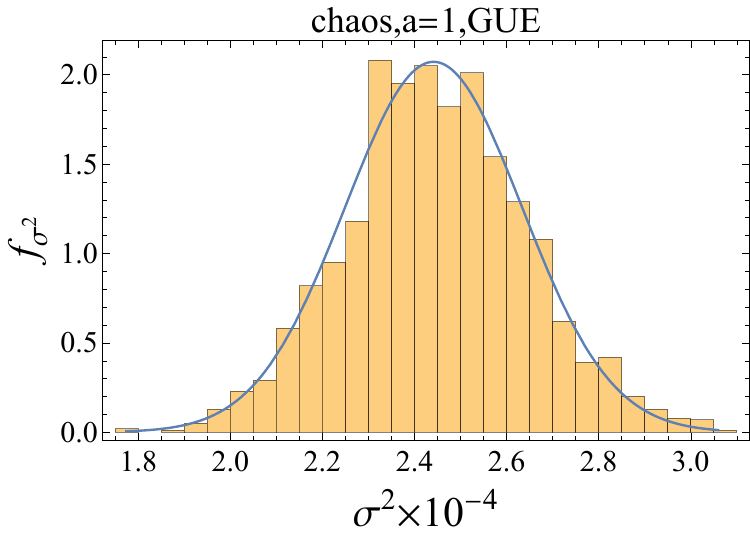}
                 \caption{Chaotic case $a=1$  of GUE. Skewness=0.0881284, Kurtosis=3.06162.}
                 \label{fig:B-GUE-fit-chaos-choice-100}
        \end{subfigure}
        \caption{Fitted normal distribution for GOE and GUE, $\mathcal{N}_{max}=100$.  The Lanczos coefficients  $\left \{ b_n \right \} $ from step-sizes $1000\le n\le 1500$ are used.}
        \label{fig:billiard-fitted-distribution-Gaussian-choice-100}
\end{figure}

\begin{figure}[h]
        \centering
        \begin{subfigure}[b]{0.48\textwidth}
                \centering
                \includegraphics[width=\linewidth]{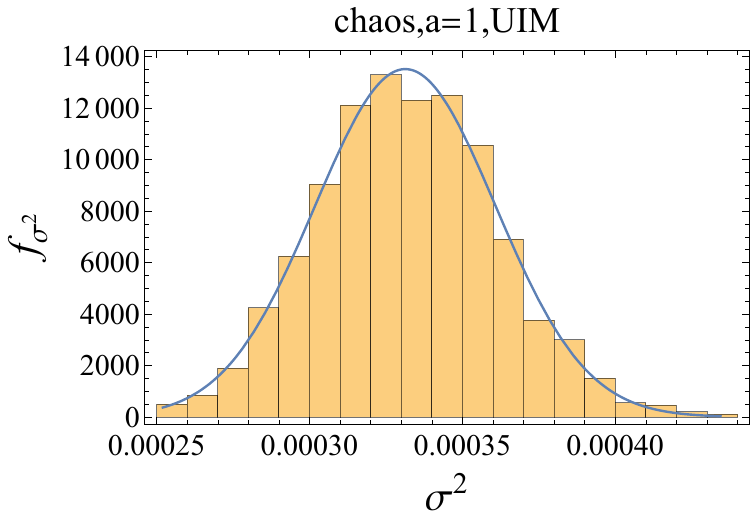}
                 \caption{Chaotic case $a=1$  of UIM. Skewness=0.160477, Kurtosis=3.02122.}
                 \label{fig:B-uni-fit-chaosIM-choice-100}
        \end{subfigure}
           \hfill
        \begin{subfigure}[b]{0.48\textwidth}
                \centering
                \includegraphics[width=\linewidth]{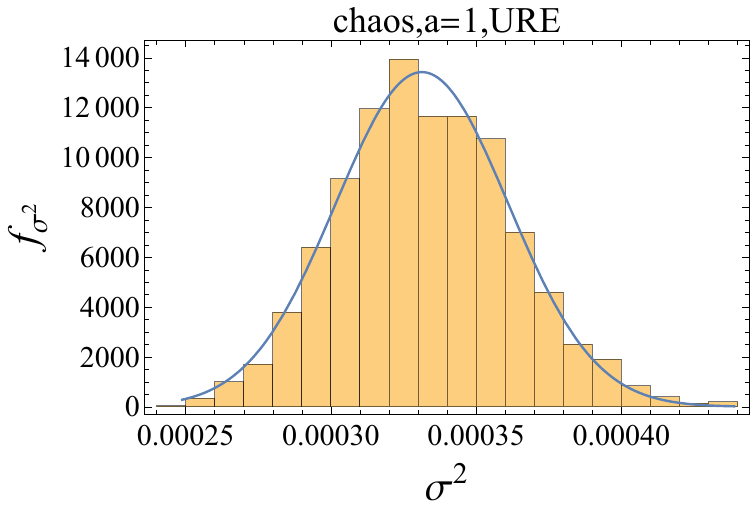}
                 \caption{Chaotic case $a=1$  of URE. Skewness=0.200908, Kurtosis=3.02255.}
                 \label{fig:B-uni-fit-chaosRE-choice-100}
        \end{subfigure}
        \begin{subfigure}[b]{0.48\textwidth}
                \centering
                \includegraphics[width=\linewidth]{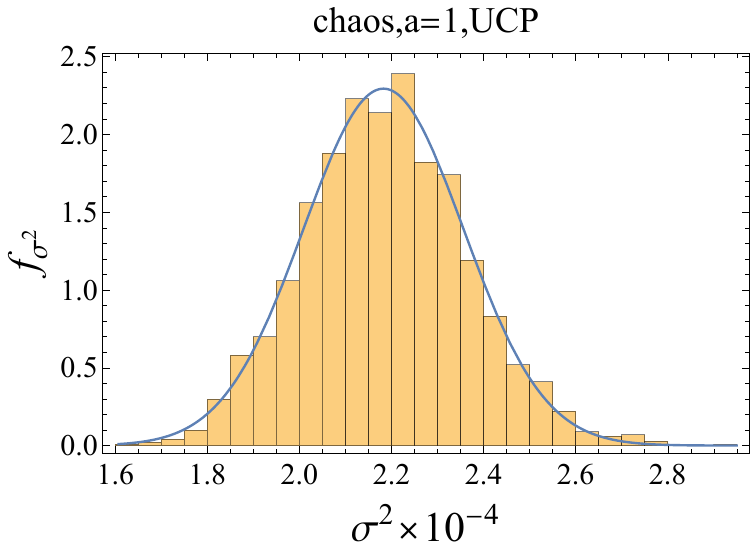}
                 \caption{Chaotic case $a=1$  of UCP. Skewness=0.181879, Kurtosis=3.18625.}
                 \label{fig:uni-fit-chaosCP-choice-100}
        \end{subfigure}
           \hfill
        \begin{subfigure}[b]{0.48\textwidth}
                \centering
                \includegraphics[width=\linewidth]{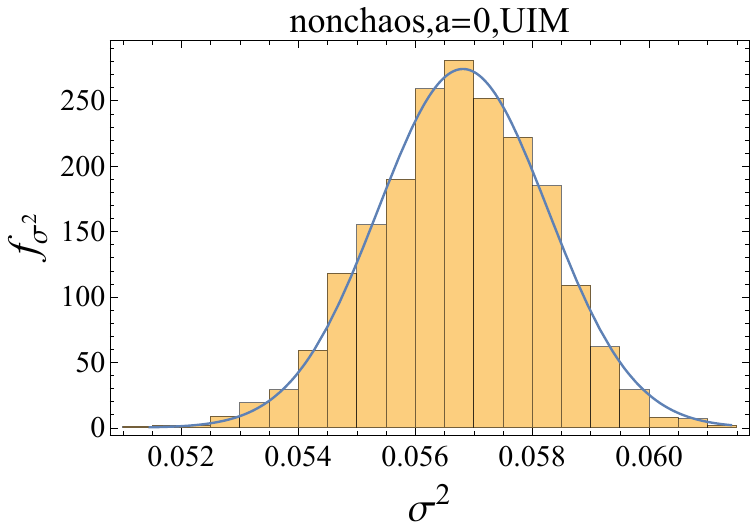}
                 \caption{Nonchaotic case $a=0$  of UIM. Skewness=-0.165951, Kurtosis=3.02278.}
                 \label{fig:uni-fit-nonchaosIM-choice-100}
        \end{subfigure}
        \begin{subfigure}[b]{0.48\textwidth}
                \centering
                \includegraphics[width=\linewidth]{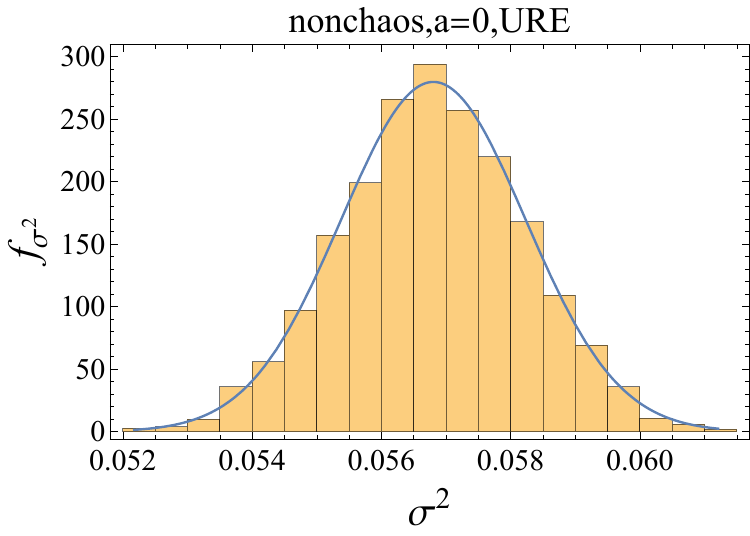}
                 \caption{Nonchaotic case $a=0$  of URE. Skewness=-0.0222373, Kurtosis=2.87503.}
                 \label{fig:uni-fit-nonchaosRE-choice-100}
        \end{subfigure}
           \hfill
        \begin{subfigure}[b]{0.48\textwidth}
                \centering
                \includegraphics[width=\linewidth]{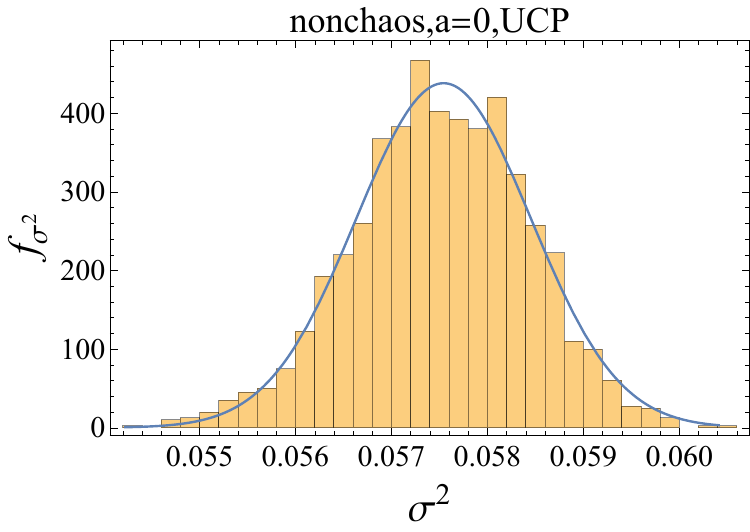}
                 \caption{Nonchaotic case $a=0$  of UCP. Skewness=-0.169296, Kurtosis=3.0196.}
                 \label{fig:uni-fit-nonchaosCP-choice-100}
        \end{subfigure}
        \caption{Fitted normal distribution for uniform distributions, $\mathcal{N}_{max}=100$. The Lanczos coefficients  $\left \{ b_n \right \} $ from step-sizes $1000\le n\le 1500$ are used.}
        \label{fig:billiard-fitted-distribution-uniform-choice-100}
\end{figure}

\clearpage
\bibliographystyle{JHEP}
\bibliography{chaos}

\providecommand{\href}[2]{#2}\begingroup\raggedright\begin{thebibliography}{10}

\bibitem{Parker:2018yvk}
D.E.~Parker, X.~Cao, A.~Avdoshkin, T.~Scaffidi and E.~Altman, \emph{{A
  Universal Operator Growth Hypothesis}},
  \href{https://doi.org/10.1103/PhysRevX.9.041017}{\emph{Phys. Rev. X}
  {\bfseries 9} (2019) 041017}
  [\href{https://arxiv.org/abs/1812.08657}{{\ttfamily 1812.08657}}].

\bibitem{Maldacena:2015waa}
J.~Maldacena, S.H.~Shenker and D.~Stanford, \emph{{A bound on chaos}},
  \href{https://doi.org/10.1007/JHEP08(2016)106}{\emph{JHEP} {\bfseries 08}
  (2016) 106} [\href{https://arxiv.org/abs/1503.01409}{{\ttfamily
  1503.01409}}].

\bibitem{1969JETP28:1200L}
A.I.~{Larkin} and Y.N.~{Ovchinnikov}, \emph{{Quasiclassical Method in the
  Theory of Superconductivity}}, {\emph{Soviet Journal of Experimental and
  Theoretical Physics} {\bfseries 28} (1969) 1200}.

\bibitem{Trigueros:2021rwj}
F.B.~Trigueros and C.-J.~Lin, \emph{{Krylov complexity of many-body
  localization: Operator localization in Krylov basis}},
  \href{https://doi.org/10.21468/SciPostPhys.13.2.037}{\emph{SciPost Phys.}
  {\bfseries 13} (2022) 037}
  [\href{https://arxiv.org/abs/2112.04722}{{\ttfamily 2112.04722}}].

\bibitem{Rabinovici:2021qqt}
E.~Rabinovici, A.~S\'anchez-Garrido, R.~Shir and J.~Sonner, \emph{{Krylov
  localization and suppression of complexity}},
  \href{https://doi.org/10.1007/JHEP03(2022)211}{\emph{JHEP} {\bfseries 03}
  (2022) 211} [\href{https://arxiv.org/abs/2112.12128}{{\ttfamily
  2112.12128}}].

\bibitem{Rabinovici:2022beu}
E.~Rabinovici, A.~S\'anchez-Garrido, R.~Shir and J.~Sonner, \emph{{Krylov
  complexity from integrability to chaos}},
  \href{https://doi.org/10.1007/JHEP07(2022)151}{\emph{JHEP} {\bfseries 07}
  (2022) 151} [\href{https://arxiv.org/abs/2207.07701}{{\ttfamily
  2207.07701}}].

\bibitem{Craps:2024suj}
B.~Craps, O.~Evnin and G.~Pascuzzi, \emph{{Multiseed Krylov complexity}},
  \href{https://arxiv.org/abs/2409.15666}{{\ttfamily 2409.15666}}.

\bibitem{Dyson:1962JMP}
F.J.~{Dyson}, \emph{{Statistical Theory of the Energy Levels of Complex
  Systems. II}}, \href{https://doi.org/10.1063/1.1703774}{\emph{Journal of
  Mathematical Physics} {\bfseries 3} (1962) 157}.

\bibitem{Gutzwiller:1971fy}
M.C.~Gutzwiller, \emph{{Periodic orbits and classical quantization
  conditions}}, \href{https://doi.org/10.1063/1.1665596}{\emph{J. Math. Phys.}
  {\bfseries 12} (1971) 343}.

\bibitem{Berry:1977RSPSA}
M.V.~{Berry} and M.~{Tabor}, \emph{{Level Clustering in the Regular Spectrum}},
  \href{https://doi.org/10.1098/rspa.1977.0140}{\emph{Proceedings of the Royal
  Society of London Series A} {\bfseries 356} (1977) 375}.

\bibitem{Bohigas:1983er}
O.~Bohigas, M.J.~Giannoni and C.~Schmit, \emph{{Characterization of chaotic
  quantum spectra and universality of level fluctuation laws}},
  \href{https://doi.org/10.1103/PhysRevLett.52.1}{\emph{Phys. Rev. Lett.}
  {\bfseries 52} (1984) 1}.

\bibitem{Blumel:1990zz}
R.~Blumel and U.~Smilansky, \emph{{Random-matrix description of chaotic
  scattering: Semiclassical approach}},
  \href{https://doi.org/10.1103/PhysRevLett.64.241}{\emph{Phys. Rev. Lett.}
  {\bfseries 64} (1990) 241}.

\bibitem{Hashimoto:2023swv}
K.~Hashimoto, K.~Murata, N.~Tanahashi and R.~Watanabe, \emph{{Krylov complexity
  and chaos in quantum mechanics}},
  \href{https://doi.org/10.1007/JHEP11(2023)040}{\emph{JHEP} {\bfseries 11}
  (2023) 040} [\href{https://arxiv.org/abs/2305.16669}{{\ttfamily
  2305.16669}}].

\bibitem{BERRY1981163}
M.~Berry, \emph{Quantizing a classically ergodic system: Sinai's billiard and
  the kkr method},
  \href{https://doi.org/https://doi.org/10.1016/0003-4916(81)90189-5}{\emph{Annals
  of Physics} {\bfseries 131} (1981) 163}.

\bibitem{benettin1978numerical}
G.~Benettin and J.-M.~Strelcyn, \emph{Numerical experiments on the free motion
  of a point mass moving in a plane convex region: Stochastic transition and
  entropy}, {\emph{Physical review A} {\bfseries 17} (1978) 773}.

\bibitem{McDonald:1979zz}
S.W.~McDonald and A.N.~Kaufman, \emph{{Spectrum and Eigenfunctions for a
  Hamiltonian with Stochastic Trajectories}},
  \href{https://doi.org/10.1103/PhysRevLett.42.1189}{\emph{Phys. Rev. Lett.}
  {\bfseries 42} (1979) 1189}.

\bibitem{Casati:1980ytd}
G.~Casati, F.~Valz-Gris and I.~Guarnieri, \emph{{On the connection between
  quantization of nonintegrable systems and statistical theory of spectra}},
  \href{https://doi.org/10.1007/BF02798790}{\emph{Lett. Nuovo Cim.} {\bfseries
  28} (1980) 279}.

\bibitem{Weinmann1995h2eOF}
Weinmann, M{\"u}ller-Groeling, Pichard and Frahm, \emph{h/2e oscillations for
  correlated electron pairs in disordered mesoscopic rings.}, {\emph{Physical
  review letters} {\bfseries 75 8} (1995) 1598}.

\bibitem{Nandy:2024htc}
P.~Nandy, A.S.~Matsoukas-Roubeas, P.~Mart\'\i{}nez-Azcona, A.~Dymarsky and
  A.~del Campo, \emph{{Quantum Dynamics in Krylov Space: Methods and
  Applications}},  \href{https://arxiv.org/abs/2405.09628}{{\ttfamily
  2405.09628}}.

\bibitem{Dymarsky:2019elm}
A.~Dymarsky and A.~Gorsky, \emph{{Quantum chaos as delocalization in Krylov
  space}}, \href{https://doi.org/10.1103/PhysRevB.102.085137}{\emph{Phys. Rev.
  B} {\bfseries 102} (2020) 085137}
  [\href{https://arxiv.org/abs/1912.12227}{{\ttfamily 1912.12227}}].

\bibitem{fleishman1977fluctuations}
L.~Fleishman and D.~Licciardello, \emph{Fluctuations and localization in one
  dimension}, {\emph{Journal of Physics C: Solid State Physics} {\bfseries 10}
  (1977) L125}.

\bibitem{Eynard:2015aea}
B.~Eynard, T.~Kimura and S.~Ribault, \emph{{Random matrices}},
  \href{https://arxiv.org/abs/1510.04430}{{\ttfamily 1510.04430}}.

\bibitem{wishart1928generalised}
J.~Wishart, \emph{The generalised product moment distribution in samples from a
  normal multivariate population}, {\emph{Biometrika} (1928) 32}.

\bibitem{mehta1991random}
M.~Mehta, \emph{Random Matrices}, Academic Press (1991).

\bibitem{haake1991quantum}
F.~Haake, \emph{Quantum signatures of chaos}, Springer (1991).

\end{thebibliography}\endgroup

\end{document}